\documentclass[12pt]{article}

\usepackage{amsmath}
\usepackage{paralist}
\usepackage[colorlinks=true]{hyperref}
\usepackage{amssymb}
\usepackage{amsthm}
\usepackage{latexsym}
\usepackage{cite}
\usepackage{psfrag}
\usepackage{color}
\usepackage{amsfonts}
\usepackage{mathrsfs}
\usepackage{url}
\hypersetup{urlcolor=blue, citecolor=red}

\makeatletter
\@addtoreset{equation}{section}
\makeatother

\newtheorem{theorem}{Theorem}[section]
\newtheorem{corollary}[theorem]{Corollary}
\newtheorem{conjecture}[theorem]{Conjecture}
\newtheorem{proposition}[theorem]{Proposition}
\newtheorem{lemma}[theorem]{Lemma}
\theoremstyle{definition}
\newtheorem{definition}[theorem]{Definition}
\newtheorem{remark}[theorem]{Remark}
\newtheorem{notation}[theorem]{Notation}
\newtheorem{example}[theorem]{Example}

\newcommand{\A}{\mathcal{A}}
\newcommand{\BB}{\mathfrak{B}}
\newcommand{\Bb}{\mathbb{B}}
\newcommand{\B}{\mathscr{B}}
\newcommand{\C}{\mathcal{C}}
\newcommand{\K}{\mathcal{K}}
\newcommand{\N}{\mathcal{N}}

\newcommand{\Nb}{\mathbb{N}}
\newcommand{\V}{\mathcal{V}}
\newcommand{\PG}{\mathrm{PG}}
\newcommand{\F}{\mathbb{F}}
\newcommand{\cb}{\mathbf{c}}
\newcommand{\e}{\mathbf{e}}
\newcommand{\vb}{\mathbf{v}}
\newcommand{\x}{\mathbf{x}}
\newcommand{\U}{\mathscr{U}}

\begin{document}

\title
{On the weight distribution of the cosets of MDS codes
\date{}
}
\maketitle

\begin{center}
{\sc Alexander A. Davydov
\footnote{A.A. Davydov ORCID \url{https://orcid.org/0000-0002-5827-4560}}
}\\
{\sc\small Institute for Information Transmission Problems (Kharkevich institute)}\\
 {\sc\small Russian Academy of Sciences}\\
 {\sc\small Moscow, 127051, Russian Federation}\\\emph {E-mail address:} adav@iitp.ru\medskip\\
 {\sc Stefano Marcugini
 \footnote{S. Marcugini ORCID \url{https://orcid.org/0000-0002-7961-0260}},
 Fernanda Pambianco
 \footnote{F. Pambianco ORCID \url{https://orcid.org/0000-0001-5476-5365}}
 }\\
 {\sc\small Department of  Mathematics  and Computer Science,  Perugia University,}\\
 {\sc\small Perugia, 06123, Italy}\\
 \emph{E-mail address:} \{stefano.marcugini, fernanda.pambianco\}@unipg.it
\end{center}

\textbf{Abstract.}
 The weight distribution of the cosets of maximum distance separable (MDS) codes is considered.
In 1990, P.G.\ Bonneau proposed a relation to obtain the full weight distribution of a coset of an  MDS code with minimum distance $d$ using  the known numbers of vectors of weights $\le d-2$ in this coset. In this paper, the Bonneau formula is transformed into a more structured and convenient form. The new version of the formula allows to consider effectively cosets of distinct weights $W$. (The weight $W$ of a coset is the smallest Hamming weight of any vector in the coset.)  For each of the considered $W$ or regions of $W$, special relations more simple than the general ones are obtained.
For the MDS code cosets of  weight $W=1$ and weight $W=d-1$ we obtain formulas of the weight distributions depending only on the code parameters.
This proves that all the cosets of weight $W=1$ (as well as $W=d-1$) have the same weight distribution. The cosets of weight $W=2$ or $W=d-2$ may have different weight distributions; in this case, we proved that the distributions are symmetrical in some sense.
The weight distributions of the cosets of MDS codes corresponding to arcs  in the projective plane $\mathrm{PG}(2,q)$ are also considered. For  MDS codes of covering radius $R=d-1$ we obtain the number of the weight $W=d-1$ cosets and their weight distribution that gives rise to a certain classification of the so-called deep holes. We show that any MDS code of covering radius $R=d-1$ is an  almost perfect multiple covering of the farthest-off points (deep holes); moreover, it  corresponds to an optimal multiple saturating set in the projective space $\mathrm{PG}(N,q)$.



\section{Introduction. The main results}
Let $\F_{q}$ be the Galois field with $q$ elements, $\F_{q}^*=\F_{q}\setminus\{0\}$. Let $\F_{q}^{n}$ be
the space of $n$-dimensional vectors over ${\mathbb{F}}_{q}$.  We denote by  $[n,k,d]_{q}R$ an $\F_q$-linear code of length $n$, dimension $k$, minimum distance~$d$, and covering radius $R$. In this notation, one may omit $R$ and call ``minimum distance''  simply ``distance''. If $d=n-k+1$, it is a maximum distance separable (MDS) code. The   generalized Reed-Solomon (GRS) codes, including generalized doubly-extended Reed-Solomon (GDRS) codes and generalized triple-extended Reed-Solomon (GTRS) codes, form an important class of MDS codes.   For an introduction to coding theory, see \cite{Blahut,HufPless,MWS,Roth}. For preliminaries, see Section \ref{subsec_MDS_Prelim}.

Let $\PG(N,q)$ be the $N$-dimensional projective space over~$\F_q$. An $n$-\emph{arc} in  $\PG(N,q)$ is a
set of $n$ points such that no $N +1$ points belong to
the same hyperplane of $\PG(N,q)$. An $n$-arc is complete if it is not contained in an $(n+1)$-arc. Arcs and MDS codes are equivalent objects, see e.g. \cite{BallLav,EtzStorm2016,LandSt,MWS}.
 For an introduction to projective spaces over finite fields and connections between  projective geometry, coding theory, and combinatorics, see \cite{Ball-book2015,BallLav,EtzStorm2016,Hirs_PGFF,HirsStor-2001,HirsThas-2015,LandSt}.

A \emph{coset} of a code is a translation of the code. A coset $\V$ of an $[n,k,d]_{q}$ code $\C$ can be represented as
\begin{align}\label{eq1_coset}
  \V=\{\x\in\F_q^n\,|\,\x=\cb+\vb,\cb\in \C\}\subset\F_q^n
\end{align}
 where $\vb\in \V$ is a vector  fixed for the given representation and $\cb$ is a codeword; see e.g. \cite{Blahut,HufPless,MWS,Roth}. The \emph{weight $W$ of a coset} is the smallest Hamming weight of any vector in the coset. For preliminaries, see Section \ref{subsec_cosets}.

 The weight distribution of the code cosets, their classification,  the number of the cosets with distinct distributions, are interesting by themselves; they are important combinatorial properties of a code. In particular, the weight distribution serves to estimate decoding results. Knowledge of the weight distributions
of the code cosets gives information on the distance distribution of the code itself. There are many papers connected with distinct aspects of the weight distribution of the code cosets, see e.g. \cite{AsmMat,Blahut,Blahut2008,BlokPelSzo,Bonneau1990,CharpHelZin,CheungIEEE1989,CheungIEEE1992,DMP_IntegrWeight1_2Cosets,DMP_CosetsRScod4,Delsarte4Fundam}, \cite[Section 6.3]{DelsarteBook}, \cite{DelsarteLeven,Helleseth}, \cite[Section 7]{HufPless}, \cite{KasLin,JurrPellik,KaipaIEEE2017}, \cite[Section 1.4.6]{Klove}, \cite{MW1963}, \cite[Sections 5.5, 6.6, 6.9]{MWS}, \cite{Schatz1980,XuXu2019,ZDK_DHRIEEE2019}, and the references therein.

In \cite{CheungIEEE1989}, for  $[n,k,d]_{q}$ MDS codes, the integrated weight distribution for the union of all cosets of weight $W\le\lfloor(d-1)/2\rfloor$ is obtained. In \cite{DMP_IntegrWeight1_2Cosets}, using results of \cite{CheungIEEE1989}, the weight distributions for the unions of all cosets of weight 1 and all cosets of weight~2 of MDS codes are given. In \cite{DMP_CosetsRScod4}, the weight distributions of all cosets  (without any unions) of the specific $[q+1,q-3,5]_q3$ GDRS code are obtained.

 Overall, in the literature, to the best of authors' knowledge, for non-binary MDS codes, sporadic or infinite families, the weight distributions of the code cosets (without any unions) are not considered, apart from the case in \cite{DMP_CosetsRScod4}.

\emph{In this paper}, some steps are made to fill this gap.

In a few works, see e.g. \cite{Bonneau1990,Delsarte4Fundam,HufPless,MWS}, it is shown that for an MDS code of distance~$d$, the weight distribution of any coset is uniquely determined if, in the coset, the numbers of vectors of weights $1,2,\ldots,d-2$ are known. Methods to obtain the weight distribution of the coset using these $d-2$ known numbers are considered in \cite{Bonneau1990}, \cite[Section 7]{HufPless}. The approach of \cite{HufPless} can be used for distinct codes, not necessary MDS.

For MDS codes, the approach of \cite{Bonneau1990} is simpler that the one of \cite{HufPless}. In Bonneau paper \cite{Bonneau1990}, it is proposed a formula which is a remarkable direct relation between the $d-2$ known numbers of vectors in a coset and the full weight distribution of this cosets, see~\eqref{eq2_wd_cosetBon}--\eqref{eq2_Bw2Bon} in Section \ref{subsec_MDS_Prelim}. It seems the method of \cite{Bonneau1990} is currently the most convenient and productive tool for obtaining the weight distributions of the cosets of MDS codes.

However, in the literature, as far as it is known to the authors, the Bonneau formula is not developed and applied. Note also that the Bonneau formula seems to be slightly ``non-transparent''; its application to specific cases not always gives impressive results. Also, its complexity of calculations is relatively  great.

\emph{In this paper}, we transform the formula of \cite{Bonneau1990} into a simpler form. The new form allowed us to make further simplifications for distinct specific cases.  Then we apply the improved tools to
obtain the weight distributions of the cosets of distinct MDS codes and obtain a number of new results.

The new version of the formula of \cite{Bonneau1990} is given in Theorem \ref{th1_BonTrans}, see Section~\ref{sec_transform}.

\begin{theorem}\label{th1_BonTrans} \textbf{(Bonneau formula transformed)}
Let $\C$ be an $[n,k,d]_q$ MDS code of distance $d\ge2$. Let $\V$ be one of its cosets. Let $A_w(\C)$ be the number of weight $w$ codewords of $\C$. Let $B_w(\V)$ be the number of weight $w$ vectors in the coset~$\V$. Assume that all the values of $B_v(\V)$ with $0\le v\le d-2$ are known.  Then, for $w\ge d-1$, the weight distribution of\/ $\V$ is as follows:
\begin{align}\label{eq1_BonGenTransf}
&B_w(\V)=A_w(\C)-\Omega_w^{(0)}(n,d)+\sum_{v=0}^{d-2}\Omega_w^{(v)}(n,d)B_v(\V),~w=d-1,d,\ldots,n,
\end{align}
where
\begin{align}
&\Omega_w^{(v)}(n,d)=(-1)^{w-d}\binom{n-v}{w-v}\binom{w-1-v}{d-2-v},~d\ge2.\label{eq1_Omega}
\end{align}
\end{theorem}

The new version of the formula seems to be more structured and more convenient for applying than the original one, also it requires less calculations, see Remarks \ref{rem3:comp} and \ref{rem4:comp} for detail.

Using the new version \eqref{eq1_BonGenTransf}, \eqref{eq1_Omega}, we consider separately cosets of weights $W$ with $W=1,d-1,d-2$, and $2\le W\le \lfloor(d-1)/2\rfloor$. Such a consideration is absent in the literature, as far as it is known to the authors. For each of these $W$ or regions of $W$ we obtain special relations simpler than \eqref{eq1_BonGenTransf}, \eqref{eq1_Omega}. In turn, the results of Section \ref{sec_spec} are widely used in Sections \ref{sec_wd2}--\ref{sec_multip}.

For the cosets of weights $W=1$ and  $W=d-1$ of an $[n,k,d]_q$ MDS code, we \emph{obtain formulas of the weight distributions depending only on the code parameters $n,d,q$.} These formulas are a new and useful result. In particular, they prove that all the MDS code cosets of weight $1$ (as well as $d-1$) have the same weight distribution, see Corollary \ref{cor4:W1d-1}.
The obtained weight distribution of the weight $d-1$ cosets is the base of  new results on multiple coverings and deep holes, see Section~\ref{sec_multip}.

The weight $d-2$ cosets may have different weight distributions; in this case, the distributions, as we prove, are symmetrical in some sense.

The weight distribution of the weight 2 cosets
 of MDS codes of distance $d\ge5$ is considered in detail, see Section \ref{sec_wd2}. We show that the  distribution  is uniquely determined by the number
  of weight $d-2$ vectors in this coset.
    Formulas for the weight distributions are obtained. A necessary condition for equality of weight distributions is proved. Again, as well as for weight $d-2$ cosets, we prove that different weight distributions of weight 2 cosets are symmetrical.

The \emph{symmetry of different weight distributions} of weight 2 and weight $d-2$ cosets of MDS codes, proved in this paper, is an unexpected and interesting new result.

Also, we consider the weight distribution of the cosets of MDS codes of distance $d=4$, see Section \ref{sec_d=3,4}. These codes are interesting for theory and practice. We use their connections with the conics and  hyperovals in the plane $\PG(2,q)$. New properties of shortened conics, needed for the coset weight distributions, are obtained.

In coding theory, \emph{farthest-off points} or \emph{deep holes}, i.e. vectors of $\F_{q}^{n}$ lying at distance~$R$ from an $[n,k,d]_{q}R$ code, play an important role. There are useful relations between the deep holes and the bounds on the size of the lists in the list decoding~ of GRS codes, see e.g. \cite[Sections 4.3,4.8]{Blahut2008}, \cite{HongWu2016,JusHoh2001,KaipaIEEE2017}, \cite[Chapter 9]{Roth}, \cite{XuXu2019,ZDK_DHRIEEE2019} and the references therein.

In this context, linear multiple covering codes, called \emph{multiple coverings of the farthest-off points \emph{(}MCF
codes for short\emph{)}} or \emph{ multiple coverings of the deep holes}, are of great interest. For an introduction to this topic, including a one-to-one correspondence between MCF codes and multiple saturating sets in the spaces $\PG(N,q)$, see \cite{BDGMP_MultCov,BDGMP_MultCovFurth,BDMP_TwistCub,CHLL_CovCodBook} and the references therein. For preliminaries with  the definitions of perfect  and almost perfect MCF codes, see  Section~\ref{subsec:multip_prelim}.

We show, see Section \ref{subsec:multip_newres}, that any $[n,k,d]_{q}R$ MDS code of maximal possible covering radius $R=d-1$ is an  almost perfect  MCF code such that for each farthest-off vector $\x\in\F_q^n$ there are exactly $\binom{n}{d-1}$ codewords at Hamming distance $d-1$ from $\x$. It is an important result, as MDS codes of covering radius $R=d-1$ are a wide class of codes, see Theorem \ref{th7_R=d-1}.

The number of codewords at distance $R$\/ from a deep hole, the weight distribution of the farthest-off cosets of weight $d-1$, and the number of such cosets and its estimate  given in Section \ref{sec_multip}, can be considered as the \emph{classification of the farthest-off vectors \emph{(}deep holes\emph{)}} of MDS codes.

Note that MDS codes of covering radius $R=d-1$ provide almost perfect  MCF codes with any covering radius $R$ and the corresponding optimal multiple $(\rho,\mu)$-saturating sets with any parameter $\rho$, see Section \ref{sec_multip}. In the literature, as far as it is known to the authors, almost perfect  MCF codes with $R>3$ and optimal multiple $(\rho,\mu)$-saturating sets with $\rho>2$ are not described.

So, we improve the known formula for obtaining the weight distributions of the cosets of MDS codes, consider these distributions for a number codes, and also obtain new results on related problems connected with  farthest-off vectors (deep holes). The results obtained extend the knowledge on cosets of MDS codes.

 The paper is organized as follows. Section \ref{sec_prelimin} contains preliminaries.
 In Section~\ref{sec_transform} we transform the original Bonneau formula \eqref{eq2_wd_cosetBon}--\eqref{eq2_Bw2Bon} to the new version \eqref{eq1_BonGenTransf}, \eqref{eq1_Omega}. In Section~\ref{sec_spec}, we consider specific cases of the weight distribution of the cosets of MDS codes. In Section~\ref{sec_wd2}, the weight distribution of the weight~2 cosets of MDS codes of distance $d\ge5$ is explored. In Section \ref{sec_d=3,4},  the weight distribution of the cosets of MDS codes of distance $d=4$ is investigated. In Section~\ref{sec_multip}, we study MDS codes as multiple coverings of the deep holes and the corresponding multiple saturating sets.

\section{Preliminaries}\label{sec_prelimin}
We introduce notations and remind known definitions and properties of linear codes, their cosets,  and the normal rational curves in the spaces $\PG(N,q)$, see  \cite{Ball-book2015,BallLav,BDMP_AlmComplArc,BartGiulPlat,Blahut,Blahut2008,Bonneau1990,CheungIEEE1989,CheungIEEE1990,CheungIEEE1992,%
DMP_IntegrWeight1_2Cosets,Delsarte4Fundam,EzerGrasSoleMDS2011,GabKl,HufPless,MWS,Roth} and the references therein.

\subsection{Cosets of a linear code}\label{subsec_cosets}
 \begin{notation}\label{not1}
 For an $[n,k,d]_{q}$ code $\C$ and its cosets $\V$ of the form \eqref{eq1_coset}, we use the following notations and definitions:
\begin{align*}
&t(\C)=\left\lfloor(d-1)/2\right\rfloor&&\text{the number of errors correctable by the code }\C;\displaybreak[3]\\
&\text{weight of a vector}&&\text{Hamming weight of the vector;}\displaybreak[3]\\
&wt(\x)&&\text{the Hamming weight of a vector  $\x\in\F_q^n$;}\displaybreak[3]\\
&\#M&&\text{the cardinality of a set }M;\displaybreak[3]\\
&A_w(\C)&&\text{the number of weight $w$ codewords of the code $\C$;}\displaybreak[3]\\
&S(\C)&&\text{the set of non-zero weights in}\,\C;S(\C)=\{w>0|A_w(\C)\ne0\};\displaybreak[3]\\
&s(\C)=\#S(\C)&&\text{the number of non-zero weights in }\C;\displaybreak[3]\\
&\C^\bot&&\text{the $[n,n-k,d^\bot]_q$ code dual to $\C$;}\displaybreak[3]\\
&\cb_w&&\text{weight $w$ codeword of $\C$};~ wt(\cb_w)=w;\displaybreak[3]\\
&B_w(\V)&&\text{the number of weight $w$ vectors in the coset }\V;\displaybreak[3]\\
&\text{a coset leader}&&\text{a vector in the coset having the smallest Hamming weight;}\displaybreak[3]\\
&\text{the weight of a coset}&&\text{the smallest Hamming weight of any vector in the coset;}\displaybreak[3]\\
&\V^{(W)}&&\text{a coset of weight }W;~~B_w(\V^{(W)})=0\text{ if }w<W;\displaybreak[3]\\
&\Nb^{(W)}_\Sigma(\C)&&\text{the total number of the cosets of weight $W$ of the code }\C;\displaybreak[3]\\
&\B_w^{\Sigma}(\V^{(W)})&&\text{the overall number of weight $w$ vectors in all the cosets}\displaybreak[3]\\
&&&\text{of weight }W;\displaybreak[3]\\
&\vb_w&&\text{weight $w$ vector of $\F_q^{n}$};~ wt(\vb_w)=w;\displaybreak[3]\\
&H(\C)&&\text{an $(n-k)\times n$ parity check matrix of }\C;\displaybreak[3]\\
&tr&&\text{sign of the transposition};\displaybreak[3]\\
&H(\C)\x^{tr}&&\text{the \emph{syndrome} of a vector }\x\in\F_q^n,~~H(\C)\x^{tr}\in \F_q^{n-k};\displaybreak[3]\\
&d(\x,\cb)&&\text{the Hamming distance between a vector $\x$ and }\displaybreak[3]\\
&&&\text{a vector $\cb$ of }\F_{q}^{n};\displaybreak[3]\\
&d(\x,\C)=\min\limits_{\cb\in \C}d(\x,\cb)&&\text{the Hamming distance between a vector }\x\in\F_q^n \text{ and }\C;\displaybreak[3]\\
&f_\delta(\x,\C)&&\text{for vector $\x\in\F_q^n$, the number of codewords $\cb\in\C$ such}\displaybreak[3]\\
&&&\text{that }d(\x,\cb)=\delta.
\end{align*}
\end{notation}

 For a coset $\V^{(W)}$ of weight $W$ the number of all coset leaders is $B_W(\V^{(W)})$.
 If $W\leq t(\C)$ we have a unique leader and $B_W(\V^{(W)})=1$. The code $\C$  is the coset of weight zero.
  If $W> t(\C)$, then $B_W(\V^{(W)})\ge1$, i.e. a  vector of minimal weight is not necessarily unique.
 Also, for an $[n,k,d]_{q}$ code $\C$ the following holds:
 \begin{align}\label{eq2:NWSigma}
  \Nb^{(W)}_\Sigma(\C)=\binom{n}{W}(q-1)^W \text{ if }W\leq t(\C)=\left\lfloor\frac{d-1}{2}\right\rfloor.
 \end{align}

 All vectors in a code coset have the same syndrome; it is called the \emph{coset syndrome}. There is a one-to-one correspondence between cosets and syndromes.
 The syndrome of $\C$ is the zero vector of $ \F_q^{n-k}$.

The covering radius of an
$[n,k,d]_{q}$ code $\C$ is the least integer $R$ such that the space $\F_{q}^{n}$ is covered by
Hamming spheres of radius $R$ centered at the codewords. Every column of $\F_{q}^{n-k}$ is equal to a linear combination of at most $R$ columns of a parity check matrix of $\C$.
 The covering radius $R$ of the code $\C$ is equal to the maximum weight of a coset of $\C$.
 For an $[n,k,d]_qR$ MDS code we have $R\le d-1$ \cite{BartGiulPlat,GabKl}.

 \begin{theorem}\label{th2_HP_coset}
\begin{description}
      \item[(i)] \cite[Theorem 7.5.2]{HufPless} Each weight $s(\C^\bot)$ coset of $\C$ has the same weight distribution.

  \item[(ii)] \cite[Theorem 7.5.2]{HufPless}, \cite[Theorem 6.20]{MWS}
  For a code $\C$, the weight distribution of any coset of weight $<s(\C^\bot)$ is uniquely determined if, in the coset, the numbers
of vectors of weights $1,2,\ldots,s(\C^\bot)-1$ are known.
\item[(iii)]  \cite{Bonneau1990,Delsarte4Fundam} For an  MDS code of distance $d$,  the weight distribution of a coset is uniquely determined if, in the coset, the numbers
of vectors of weights $1,2,\ldots,d-2$ are known.
\end{description}
\end{theorem}

Note that using Theorem \ref{th2_WD_MDS} of Section \ref{subsec_MDS_Prelim}, one can show that, in Theorem \ref{th2_HP_coset}, the point (iii) is a special case of (ii).

For a fixed~$W$, we call the set $\{\B_w^{\Sigma}(\V^{(W)})|w=0,1,\ldots,n\}$
\emph{integral weight spectrum} of the code cosets of weight $W$, see \cite{CheungIEEE1989,CheungIEEE1992,DMP_IntegrWeight1_2Cosets} and the references therein.

The integral weight spectrum matters,  in particular, for  estimates of the so-called ``\emph{bounded distance decoder}'' \cite[Section~4.3]{Blahut2008},  \cite{CheungIEEE1989,CheungIEEE1992}.

The decoder acts as follows. For an $[n,k,d]_{q}$ code $\C$, let $\cb\in \C$ be a sent word, $\e$ be an
error vector, and $\x=\cb+\e$ be the received word. Let $\tau$ be some given distance. If $\x$ belongs to a coset of weight $>\tau$, the decoder declares the detection of an uncorrectable error. If $\x$ belongs  to a coset of weight $\le \tau$, the decoder puts the coset leader as an error vector $\e$, if the leader is unique, otherwise it forms a list of all the leaders as possible errors. This leads to incorrect decoding if, in fact, $wt(\e)>\tau$. Thus, all vectors of weight $>\tau$ belonging to the cosets of weight $\le \tau$ (resp. $> \tau$) give rise to incorrect decoding (resp. to detection of an uncorrectable error). If $\tau=t(\C)=\lfloor(d-1)/2\rfloor$, the decoder is called the \emph{decoder up to half of minimum distance}.
For the corresponding probabilities, see  \cite{CheungIEEE1989,CheungIEEE1992}.

\subsection{MDS, GDRS, GTRS, and Hamming codes; normal rational curves}\label{subsec_MDS_Prelim}

\begin{theorem}\label{th2_WD_MDS}
 \cite[Theorems 6, 10]{EzerGrasSoleMDS2011} Any $[n\le q,k,n-k+1]_q$ MDS code  has $k$ nonzero weights.
  A $[q+1,k,q+2-k]_q$ MDS code with $k\ne2$ has $k$ nonzero weights.
\end{theorem}

\begin{theorem}\label{th2_Bon} \textbf{(Bonneau formula)} \cite{Bonneau1990}
Let $\C$ be an $[n,k,d=n-k+1]_q$ MDS code of distance $d\ge2$. Let $\V$ be one of its cosets.  Assume that all the values of $B_v(\V)$ with $0\le v\le d-2$ are known. Then, for $w\ge d-1$, the weight distribution of\/ $\V$ is as follows:
\begin{align}
 &B_{w}(\V)=\BB_{w,1}(\V)+\BB_{w,2}(\V),~w\ge d-1,~d\ge2,\displaybreak[3]\label{eq2_wd_cosetBon}
 \end{align}
 where
 \begin{align}
& \BB_{w,1}(\V)=\binom{n}{w}\sum_{j=0}^{w-d+1}(-1)^j\binom{w}{j}q^{w-d+1-j},\displaybreak[3]\label{eq2_Bw1Bon}\\
&\BB_{w,2}(\V)=\sum_{j=w-d+2}^w(-1)^j\sum_{v=0}^{w-j}\binom{j+n-w}{j}\binom{n-v}{w-j-v}B_v(\V).\label{eq2_Bw2Bon}
\end{align}
\end{theorem}

For $w\ge d$, the weight distribution  of an $[n,k,d=n-k+1]_q$ MDS code $\C$ has the following form, see e.g. \cite[Theorem 7.4.1]{HufPless}, \cite[Theorem 11.3.6]{MWS}:
\begin{align}\label{eq2_wd_MDS}
 A_{w}(\C)=\binom{n}{w}\sum_{j=0}^{w-d}(-1)^j\binom{w}{j}(q^{w-d+1-j}-1).
\end{align}

For $q$ odd and even, a $(d-1)\times (q+1)$ parity check matrix $H_d$ of the $[q+1,q+2-d,d]_q$ GDRS code with $d\ge3$ can be represented \cite[Section 5.1]{Roth} as
\begin{align}\label{eq2_HDRS}
&H_d=\left[ \begin{array}{@{}cccccc@{}}
        1        &1        &\ldots&1&1        &0      \\
        \alpha_1      &\alpha_2      &\ldots&\alpha_{q-1}&0      &0     \smallskip \\
        \alpha_1^2    &\alpha_2^2    &\ldots&\alpha_{q-1}^2&0    &0     \\
        \ldots   &\ldots   &\ldots&\ldots   &\ldots&\ldots\\
        \alpha_1^{d-3}&\alpha_2^{d-3}&\ldots&\alpha_{q-1}^{d-3}&0&0     \smallskip  \\
        \alpha_1^{d-2}&\alpha_2^{d-2}&\ldots&\alpha_{q-1}^{d-2}&0&1       \\
       \end{array}\right]
       \left[\begin{array}{@{}ccccc@{}}
         v_1 & 0 & \ldots & 0 & 0 \\
         0 & v_2 & \ldots & 0 & 0 \\
        \ldots & \ldots & \ldots & \ldots & \ldots \\
         0 & 0 & \ldots& v_{q} & 0 \\
         0 & 0 & \ldots & 0 & v_{q+1}
                \end{array}\right],
\end{align}
where $\alpha_i\in\F_{q}^*$, $\alpha_i\ne \alpha_j$ if $i\ne j$, $v_i\in\F_q^*$, the elements $v_i$ do not have to be distinct.

If, from $H_d$, we remove the column $[0,0,\ldots,v_{q+1}]^{tr}$,  we obtain a parity check matrix of the \emph{singly-extended } $[q,q+1-d,d]_qR$ GRS code with $R=d-1$. If we remove the columns $[0,0,\ldots,v_{q+1}]^{tr}$, $[v_{q},0,\ldots,0]^{tr}$, and also $\delta\ge0$ other columns, we obtain a parity check  matrix of a $[q-1-\delta,q-d-\delta,d]_qR$ GRS code, $R=d-1$.

For $q$ even, a $3\times (q+2)$ parity check matrix $\widetilde{H_4}$ of the $[q+2,q-1,4]_q2$ GTRS code can be represented as
\begin{align}\label{eq2_HTRS}
&\widetilde{H_4}=\left[ \begin{array}{ccccccc}
        1        &1        &\ldots&1&1        &0 &0     \\
        \alpha_1      &\alpha_2      &\ldots&\alpha_{q-1} &0     &0 &1    \smallskip \\
        \alpha_1^2    &\alpha_2^2    &\ldots&\alpha_{q-1}^2&0    &1 &0    \\
            \end{array}\right] \left[\begin{array}{@{}ccccc@{}}
         v_1 & 0 & \ldots & 0 & 0 \\
         0 & v_2 & \ldots & 0 & 0 \\
        \ldots & \ldots & \ldots & \ldots & \ldots \\
         0 & 0 & \ldots& v_{q+1} & 0 \\
         0 & 0 & \ldots & 0 & v_{q+2}
                \end{array}\right],
\end{align}
where $\alpha_i\in\F_q^*$, $\alpha_i\ne \alpha_j$, if $i\ne j$, $v_i\in\F_q^*$, the elements $v_i$ do not have to be distinct.

The $[\frac{q^m-1}{q-1},\frac{q^m-1}{q-1}-m,3]_q1$ Hamming code is well known \cite{HufPless,MWS,Roth}. For $m=2$, it is the $[q+1,q-1,3]_q1$ GDRS code with the parity check matrix $H_3$ \eqref{eq2_HDRS}.

In $\PG(N,q)$, $2\le N\le q-2$, a \emph{normal rational curve} is any $(q+1)$-arc projectively equivalent to the arc
$\{(1,t,t^2,\ldots, t^N):t\in \F_q\}\cup \{(0,\ldots,0 ,1)\}$ where $(1,t,\ldots,t^N)$ and $(0,\ldots,0 ,1)$ are points in homogeneous coordinates.  The points  of a normal rational curve in $\PG(N,q)$
treated as columns define a parity check matrix of a $[q + 1,q-N,d=N + 2]_qR$ GDRS code \cite{EtzStorm2016,LandSt,Roth}.

If in $\PG(N,q)$, the normal rational curve is a \emph{complete $(q+1)$-arc}, then the corresponding $[q + 1,q-N,d=N + 2]_qR$ GDRS code $\C$ cannot be extended to a $[q + 2,q-N+1,d=N + 2]_q$ MDS code, i.e. $\C$ has covering radius $R=d-2$.

The following conjectures are well known.

\begin{conjecture} \label{conj2_NRC}Let $2\le N\le q-2$. In $\PG(N,q)$, every normal rational curve is a complete $(q+1)$-arc except for the cases when  $q$ is even and $N\in\{2,q-2\}$, in which one point can be added to the curve.
\end{conjecture}
\begin{conjecture} \label{conj2_MDS}\textbf{(MDS conjecture)} Let $2\le N\le q-2$. An  $[n,n-N-1,N+2]_q$ MDS code \emph{(}or, equivalently, an $n$-arc in $\PG(N,q)$\emph{)} has length $n\le q+1$ except for the cases when $q$ is even and $N\in\{2,q-2\}$, in which $n\le q+2$.
\end{conjecture}
 If the MDS conjecture holds for some pair $(N,q)$ then Conjecture \ref{conj2_NRC} holds too, but in general, the reverse is not true.
For the pairs $(N,q)$ for which MDS conjecture is proved, see
 \cite{BallLav} and the references therein.
 For the pairs $(N,q)$ for which Conjecture \ref{conj2_NRC} is proved, see \cite{BDMP_AlmComplArc}  and the references therein including \cite{St_1992_ComplNRC}.

 Thus, the $[q+1, q+2-d,d]_qR$ GDRS codes with $R=d-2$ are an important case of MDS codes of covering radius $R=d-2$. For $n-d=3$ and $n<q+1$ there exist several examples of $[n,n-3,4]_q2$ codes described as complete arcs in the projective plane $\PG(2,q)$, see e.g. \cite{Hirs_PGFF,HirsStor-2001,HirsThas-2015} and the references therein. For $n-d>3$ and $n<q+1$, we also have  MDS
codes with $R=d-2$ but in this case the number of examples is smaller, see e.g. \cite{BartGiulPlat} and the references therein.

\section{Transformation of the Bonneau formula for the weight distribution of the cosets of MDS codes}\label{sec_transform}
We do technical transformations of the parts of the Bonneau formula \eqref{eq2_wd_cosetBon}--\eqref{eq2_Bw2Bon} and obtain its new version
\eqref{eq1_BonGenTransf}, \eqref{eq1_Omega} of Theorem~\ref{th1_BonTrans}.

Throughout the paper we put that a sum $\sum_{i=A}^B\ldots$ is equal to zero if $B<A$.

We use the following combinatorial identities \cite[Section 1, Equations (I), (III), (IV), Problem 9(a),(b)]{Riordan}:
\begin{align}
 &\binom{h}{\ell}=\binom{h}{h-\ell}=\binom{h-1}{\ell}+\binom{h-1}{\ell-1},\label{eq3_Riordan_ident0}\displaybreak[3]\\
 &\binom{h}{m}\binom{m}{p}=\binom{h}{p}\binom{h-p}{m-p}=\binom{h}{m-p}\binom{h-m+p}{p}, \label{eq3_Riordan_ident2}\displaybreak[3]\\ &\sum_{j=0}^m(-1)^j\binom{h}{j}=(-1)^m\binom{h-1}{m}=(-1)^m\binom{h-1}{h-1-m},\label{eq3_Riordan_ident3}\displaybreak[3]\\
 &\sum_{j=0}^{h-m}(-1)^j\binom{h}{m+j}=\binom{h-1}{m-1}.\label{eq3_Riordan_ident4}
\end{align}

\begin{lemma}\label{lem3_Bw1}
Let $\C$ be an $[n,k,d]_q$ MDS code. Let $\V$ be one of its cosets. Let $\BB_{w,1}(\V)$ be as in \eqref{eq2_Bw1Bon}. Then
\begin{align}\label{eq3_Bw1}
 \BB_{w,1}(\V)=A_w(\C)-\Omega_w^{(0)}(n,d),~w\ge d-1.
\end{align}
\end{lemma}

\begin{proof}
We write $\BB_{w,1}(\V)$ of \eqref{eq2_Bw1Bon} as
\begin{align*}
&\binom{n}{w}\sum_{j=0}^{w-d}(-1)^j\binom{w}{j}\left(q^{w+1-d-j}-1+1\right)+
\binom{n}{w}\cdot(-1)^{w-d+1}\binom{w}{w-d+1}\displaybreak[3]\\
&=\binom{n}{w}\left(\sum_{j=0}^{w-d}(-1)^j\binom{w}{j}\left(q^{w+1-d-j}-1\right)+\sum_{j=0}^{w-d}(-1)^j\binom{w}{j}-(-1)^{w-d}\binom{w}{d-1}\right)
\end{align*}
where we applied \eqref{eq3_Riordan_ident0} to the last summand. Now we use \eqref{eq2_wd_MDS} and apply the 2-nd equality of \eqref{eq3_Riordan_ident3} to $\sum_{j=0}^{w-d}(-1)^j\binom{w}{j}$. As a result, we have
\begin{align*}
&\BB_{w,1}(\V)=A_w(\C)-(-1)^{w-d}\binom{n}{w}\left(\binom{w}{d-1}-\binom{w-1}{d-1}\right).
\end{align*}
After applying \eqref{eq3_Riordan_ident0} to $\binom{w}{d-1}-\binom{w-1}{d-1}$, the assertion follows from the notation \eqref{eq1_Omega}.
\end{proof}

\begin{lemma}\label{lem3_Bw2}
Let $\C$ be an $[n,k,d]_q$ MDS code. Let $\V$ be one of its cosets. Let $\BB_{w,2}(\V)$  be as in \eqref{eq2_Bw2Bon}.  Then
\begin{align}\label{eq3_Bw2}
 \BB_{w,2}(\V)=\sum_{v=0}^{d-2}\Omega_w^{(v)}(n,d)B_v(\V),~w\ge d-1.
\end{align}
\end{lemma}

\begin{proof} We write $\BB_{w,2}(\V)$ of \eqref{eq2_Bw2Bon} as
\begin{align}\label{eq3:Bw2a}
 \BB_{w,2}(\V)=\sum_{v=0}^{d-2}\widehat{\Omega}_w^{(v)}(n,d)B_v(\V)
\end{align}
where $\widehat{\Omega}_w^{(v)}(n,d)B_v(\V)$ is the contribution of the value $B_v(\V)$ to the final result~$B_w(\V)$.
Summing up all terms in \eqref{eq2_Bw2Bon} multiplied by $B_v(\V)$, we obtain
\begin{align}\label{eq3:Bw2b}
  \widehat{\Omega}_w^{(v)}(n,d)=\sum_{j=w-d+2}^{w-v}(-1)^j\binom{j+n-w}{j}\binom{n-v}{w-j-v},~v=0,1,\ldots,d-2.
\end{align}
By \eqref{eq3_Riordan_ident0} and the 2-nd equality of \eqref{eq3_Riordan_ident2},  we have
\begin{align*}
&\binom{j+n-w}{j}\binom{n-v}{w-j-v}=\binom{n-v}{j+n-w}\binom{j+n-w}{j}=\binom{n-v}{n-w}\binom{w-v}{j};\displaybreak[3]\\
&\widehat{\Omega}_w^{(v)}(n,d)=\binom{n-v}{n-w}\sum_{j=w-d+2}^{w-v}(-1)^j\binom{w-v}{j}.
\end{align*}
We change the variable $i=j-(w-d+2)$ and obtain
\begin{align*}
&\widehat{\Omega}_w^{(v)}(n,d)=(-1)^{w-d}\binom{n-v}{n-w}\sum_{i=0}^{d-2-v}(-1)^{i}\binom{w-v}{i+w-d+2}\\
&=(-1)^{w-d}\binom{n-v}{n-w}\binom{w-v-1}{d-2-v}=\Omega_w^{(v)}(n,d)
\end{align*}
where $\sum_{i=0}^{d-2-v}(-1)^{i}\binom{w-v}{i+w-d+2}=\binom{w-v-1}{w-d+1}$ due to  \eqref{eq3_Riordan_ident4}.
\end{proof}

The technical Lemmas \ref{lem3_Bw1} and \ref{lem3_Bw2} allow us to prove the main result of this section.\medskip

\noindent\emph{Proof of Theorem }\ref{th1_BonTrans}. Substituting \eqref{eq3_Bw1} and \eqref{eq3_Bw2} to  \eqref{eq2_wd_cosetBon} we obtain the new version \eqref{eq1_BonGenTransf}, \eqref{eq1_Omega} of the Bonneau formula and thereby prove Theorem \ref{th1_BonTrans}. \qed

\begin{remark}\label{rem3:comp}
Comparing the original Bonneau formula \eqref{eq2_wd_cosetBon}--\eqref{eq2_Bw2Bon} of \cite{Bonneau1990} with the new version~\eqref{eq1_BonGenTransf}, \eqref{eq1_Omega} of Theorem \ref{th1_BonTrans}, we see that the double sum of~\eqref{eq2_Bw2Bon} is transformed to a single one in \eqref{eq1_BonGenTransf} with regular combinatorial coefficients $\Omega_w^{(v)}(n,d)$ for $B_v(\V)$. This allows us to consider effectively cosets of distinct weights~$W$, see Section \ref{sec_spec} and Remark \ref{rem4:comp} in it. Also, in \eqref{eq1_BonGenTransf}, \eqref{eq1_Omega}, the contribution of every value $B_v(\V)$ to the final result $B_w(\V)$ is written directly and clearly, whereas, in \eqref{eq2_Bw2Bon}, its representation is relatively complicated.  This is important for theoretical investigations of the coset weight distributions.

The new version of the Bonneau formula has lesser complexity.  By \eqref{eq1_BonGenTransf}, \eqref{eq1_Omega} and \eqref{eq2_wd_cosetBon}, \eqref{eq3:Bw2a}, \eqref{eq3:Bw2b}, the contribution of the value $B_v(\V)$ to the final result~$B_w(\V)$ can be calculated in the new version (where it is $\Omega_w^{(v)}(n,d)B_v(\V)$) essentially simpler than in the original one (where it is $\widehat{\Omega}_w^{(v)}(n,d)B_v(\V)$). For confirmation, it is sufficient to compare \eqref{eq1_Omega} and \eqref{eq3:Bw2b}. Thus, the new version needs  smaller calculations than the original one.

In addition, the number $A_w(\C)$ of codewords of the given weight is now directly included in the relation for the corresponding value $B_w(\V)$ for a coset. This has a methodological  significance for research showing how the weight distributions of  a code and its cosets differ from each other. Also, distinct methods of calculation, estimate, and approximation of $A_w(\C)$ are considered in the literature, see e.g. \cite{CheungIEEE1990}, \cite[Section 1.4.8]{Klove}, and included to the systems of symbol calculations. The approximation can be useful for the estimate probabilities of decoding results, see e.g. \cite{CheungIEEE1990,CheungIEEE1992}.
\end{remark}

\section{Specific cases of the weight distribution of the cosets of MDS codes}\label{sec_spec}
We consider specific applications of the new version  \eqref{eq1_BonGenTransf} of the Bonneau formula. Using particularities of the cases considered, we simplify the formula \eqref{eq1_BonGenTransf} for them.

We note some particular cases of  \eqref{eq1_Omega}:
\begin{align}
&\Omega_{d-1}^{(v)}(n,d)=-\binom{n-j}{d-1-j},~0\le v\le d-2;\displaybreak[3]\label{eq4_Omega4}\\
&\Omega_w^{(d-2)}(n,d)=(-1)^{w-d}\binom{n-d+2}{w-d+2}=(-1)^{w-d}\binom{n-d+2}{n-w}\text{ if }w\ge d-1.\label{eq4_Omega5}
\end{align}

\subsection{The cosets of weights $1$ and $d-1$ for MDS codes of distance $d\ge2$}

\begin{theorem}\label{th4_wd1cosetMDS}
 Let  $\C$ be an $[n,k,d]_q$ MDS code with $d\ge3$. Then all its $n(q-1)$ cosets $\V^{(1)}$ of weight $1$ have the same weight distribution $B_w(\V^{(1)})$ of the form:
  \begin{align}\label{eq4_wd1wc}
&B_w(\V^{(1)})=0 \text{ if }  w\in\{0,1,\ldots,d-2\}\setminus\{1\};~B_1(\V^{(1)})=1,\displaybreak[3]\\
&B_{d-1}(\V^{(1)})=\binom{n-1}{d-1};\,B_w(\V^{(1)})=A_{w}(\C)-\Omega_w^{(0)}(n,d)+\Omega_w^{(1)}(n,d),\, w=d,\ldots,n.\notag
\end{align}
\end{theorem}

\begin{proof}
The total number $\Nb^{(1)}_\Sigma(\C)=n(q-1)$ of  weight $1$ cosets follows from \eqref{eq2:NWSigma}.  By the hypothesis, $1\in\{0,1,\ldots,d-2\}$. Then the 1-st equality of~\eqref{eq4_wd1wc} follows from the definition of weight of a coset. Also, $B_1(\V^{(1)})=1$ as the coset $\V^{(1)}$ has a unique leader. Thus, in \eqref{eq1_BonGenTransf} we have
 \begin{align*}
   \sum_{v=0}^{d-2}\Omega_w^{(v)}(n,d)B_v(\V^{(1)})=\Omega_{w}^{(1)}(n,d)B_1(\V^{(1)})=\Omega_{w}^{(1)}(n,d)
 \end{align*}
 that provides the last relation of \eqref{eq4_wd1wc} for $ w=d-1,d,d+1,\ldots,n$. Then, for $w=d-1$, we take $A_{w}(\C)=0$, use  \eqref{eq4_Omega4}, and apply \eqref{eq3_Riordan_ident0}. As a result,
 \begin{align*}
&B_{d-1}(\V^{(1)})=-\Omega_{d-1}^{(0)}(n,d)+\Omega_{d-1}^{(1)}(n,d)=\binom{n}{d-1}-\binom{n-1}{d-2}= \binom{n-1}{d-1}. \qedhere
 \end{align*}
\end{proof}

Theorem \ref{th4_wd1cosetMDS} is used in Corollary \ref{cor4:W1d-1} and Theorem \ref{th6:general} which is basic  in Section~\ref{sec_d=3,4}.

For an $[n,k,d]_qR$ MDS code, the weight $d-1$ cosets exist if $R=d-1$.
\begin{theorem}\label{th4_W=d-1}
Let $\C$ be an $[n,k,d]_qR$ MDS code of distance $d\ge2$ and covering radius $R=d-1$.  Then all its cosets $\V^{(d-1)}$ of weight $d-1$ have the weight distribution $B_w(\V^{(d-1)})$ of the form:
\begin{align}\label{eq4_Bon_w=d-1}
&B_w(\V^{(d-1)})=0 \text{ if } w=0,1,\ldots,d-2;~B_{d-1}(\V^{(d-1)})=\binom{n}{d-1};\displaybreak[3]\\
&B_w(\V^{(d-1)})=A_w(\C)-\Omega_w^{(0)}(n,d)\text{ if } w=d,d+1,\ldots,n.\notag
\end{align}
\end{theorem}

\begin{proof}
 In \eqref{eq1_BonGenTransf},
   $B_v(\V^{(d-1)})=0 $ if $0\le v\le d-2$ that provides the last relation of~\eqref{eq4_Bon_w=d-1} for $ w=d-1,d,\ldots,n$. If $w=d-1$, we use \eqref{eq4_Omega4} and take $A_{d-1}(\C)=0$.
\end{proof}

Theorem \ref{th4_W=d-1} is used in Corollary \ref{cor4:W1d-1}  and Theorems~\ref{th6:general} and \ref{th7_R=d-1} which are the basic ones in Sections \ref{sec_d=3,4} and \ref{sec_multip}, respectively.
The results of Theorem \ref{th4_W=d-1} matter because the weight $d-1$ cosets consist of the farthest-off vectors (deep holes) of an MDS code. The deep holes are important for investigations of decoding \cite{HongWu2016,JusHoh2001,KaipaIEEE2017,Roth,XuXu2019,ZDK_DHRIEEE2019} and multiple coverings \cite{BDGMP_MultCov,BDGMP_MultCovFurth,BDMP_TwistCub,CHLL_CovCodBook}, see Introduction and Section \ref{sec_multip} for details.

Corollary \ref{cor4:W1d-1} accumulates the results of Theorems \ref{th4_wd1cosetMDS}, \ref{th4_W=d-1}.

\begin{corollary}\label{cor4:W1d-1}
\begin{description}
  \item[(i)]For an $[n,k,d]_q$ MDS code  of distance $d\ge2$, all cosets of weight~$1$  have the same weight distribution of the form \eqref{eq4_Bon_w=d-1}, if $d=2$, and~\eqref{eq4_wd1wc}, if $d\ge3$. Also, the total number of the weight $1$ cosets is equal to $q-1$, if $d=2$, and to $n(q-1)$, if $d\ge3$.

  \item[(ii)] For an $[n,k,d]_q$ MDS code  of distance $d\ge2$  and covering radius $R=d-1$, all cosets of weight~$d-1$  have the same weight distribution of the form \eqref{eq4_Bon_w=d-1}.
\end{description}
\end{corollary}

\begin{proof}
  The assertions follow from Theorems \ref{th4_wd1cosetMDS} and \ref{th4_W=d-1}. Note also that for an $[n,n-1,2]_q1$ MDS code, in total, there are $q$ cosets one of which is the code itself while the remaining $q-1$ ones have weight 1.
\end{proof}

Note that by the context of the known results on MDS codes, the equality of the weight distributions of the MDS code cosets of weight 1 and $d-1$ is expected. But the \emph{formulas of the weight distributions} for these cases are new and useful, see notes after Theorems \ref{th4_wd1cosetMDS} and \ref{th4_W=d-1} on their using.

\subsection{The cosets $\V^{(W)}$ of weight $2\le W\le \left\lfloor(d-1)/2\right\rfloor$ for MDS codes of distance $d\ge5$}

\begin{theorem}\label{th4_W<=t}
Let $\C$ be an $[n,k,d]_q$ MDS code of distance $d\ge5$. Let $\V^{(W)}$ be one of its cosets of weight $W$ with $2\le W\le\left\lfloor(d-1)/2\right\rfloor$.  Assume that all the values of $B_v(\V^{(W)})$ with $d-W\le v\le d-2$ are given. Then the weight distribution $B_w(\V^{(W)})$ of\/ $\V^{(W)}$ is as follows:
\begin{align}\label{eq4_Bon_w<=t}
&B_w(\V^{(W)})=0 \text{ if } w\in\{0,1,\ldots,d-W-1\}\setminus\{W\},~B_W(\V^{(W)})=1,\displaybreak[3]\\
&B_w(\V^{(W)})=A_w(\C)-\Omega_w^{(0)}(n,d)+\Omega_w^{(W)}(n,d)+
\sum_{v=d-W}^{d-2}\Omega_w^{(v)}(n,d)B_v(\V^{(W)})\notag\\
&\text{if }w=d-1,d,\ldots,n.\notag
\end{align}
\end{theorem}

\begin{proof}
By the hypothesis, $W\in\{2,3,\ldots,d-W-1\}$. Then the 1-st equality of~\eqref{eq4_Bon_w<=t} follows from the definition of weight of a coset.   Also, $B_W(\V^{(W)})=1$ as the coset $\V^{(W)}$ has a unique leader.  Now the last equality follows from \eqref{eq1_BonGenTransf}.
\end{proof}

Theorem \ref{th4_W<=t} is used in Theorem \ref{th5:W=2} which is the basic one in Section \ref{sec_wd2}.

The weight distribution of the weight $W\le \lfloor(d-1)/2\rfloor$ cosets is useful for  estimates of the bounded distance decoder, including one up to half of minimum distance, see Section~\ref{subsec_cosets}.

\subsection{The cosets $\V^{(d-2)}$ of weight $d-2$ for MDS codes of distance $d\ge4$}
\begin{theorem}\label{th4_W=d-2}
Let $\C$ be an $[n,k,d]_q$ MDS code of distance $d\ge4$. Let $\V^{(d-2)}$ be one of its cosets of weight $d-2$.  Assume that the value $B_{d-2}(\V^{(d-2)})$ is given. Then the weight distribution $B_w(\V^{(d-2)})$ of\/ $\V^{(d-2)}$ is as follows:
\begin{align}\label{eq4_Bon_w=d-2}
&B_w(\V^{(d-2)})=0 \text{ if } w=0,1,\ldots,d-3;\displaybreak[3]\\
&B_w(\V^{(d-2)})=A_w(\C)-\Omega_w^{(0)}(n,d)+(-1)^{w-d}\binom{n-d+2}{n-w}B_{d-2}(\V^{(d-2)})\displaybreak[3]\notag\\
&\text{if }w=d-1,d,\ldots,n.\notag
\end{align}
\end{theorem}

\begin{proof}
The 1-st equality of~\eqref{eq4_Bon_w=d-2} follows from the definition of weight of a coset.  Also, for $\V^{(d-2)}$, the only non-zero term of the sum $\sum_{v=0}^{d-2}\ldots$ in \eqref{eq1_BonGenTransf} is\\
   $\Omega_w^{(d-2)}(n,d)B_{d-2}(\V^{(d-2)})$. Then we use \eqref{eq4_Omega5}.
\end{proof}

Theorem \ref{th4_W=d-2} is used in Theorems \ref{th4_symmetry_d-2} and \ref{th6:general}. In turn, Theorem \ref{th4_symmetry_d-2} proves the symmetry of different weight distributions while Theorem~\ref{th6:general} is the basic one in Section \ref{sec_d=3,4}.
Then the results of Section \ref{sec_d=3,4} are used in Section \ref{sec_multip}.

For codes of covering radius $R=d-2$, the weight $d-2$ cosets consist of the farthest-off vectors (deep holes), see Introduction and Section \ref{sec_multip} for this topic.  GDRS and  GTRS codes and codes of \cite{BartGiulPlat} are examples of MDS codes with $R=d-2$.

\begin{theorem}\label{th4_symmetry_d-2}\textbf{(symmetry of different weight distributions)}
    Let $\C$ be an $[n,k,d]_q$ MDS code with $d\ge4$. Let $\V^{(d-2)}_a$ and $\V^{(d-2)}_b$ be two of its weight $d-2$ cosets with  different weight distributions. Then independently of the values of $B_{d-2}(\V^{(d-2)}_a)$ and $B_{d-2}(\V^{(d-2)}_b)$, there is   the following symmetry of the weight distributions:
\begin{align}\label{eq4_symmetry_d-2}
&(-1)^{n+d} B_w(\V_a^{(d-2)})-B_{n+d-2-w}(V_a^{(d-2)})\displaybreak[3]\\
&=(-1)^{n+d}B_w(V_b^{(d-2)})-B_{n+d-2-w}(V_b^{(d-2)}),~w=d-1,d,\ldots,n.\notag
\end{align}
\end{theorem}

\begin{proof}
By \eqref{eq4_Bon_w=d-2} and \eqref{eq3_Riordan_ident0}, we have
\begin{align*}
&  B_w(\V^{(d-2)}_a)-B_w(\V^{(d-2)}_b)\displaybreak[3]\\
&=(-1)^{w-d}\binom{n-d+2}{n-w}\left[B_{d-2}(\V^{(d-2)}_a)-B_{d-2}(\V^{(d-2)}_b)\right];\displaybreak[3]\\
&B_{n+d-2-w}(\V^{(d-2)}_a)-B_{n+d-2-w}(\V^{(d-2)}_b)\displaybreak[3]\notag\\
&=(-1)^{n-2-w}\binom{n-d+2}{-d+2+w}\left[B_{d-2}(\V^{(d-2)}_a)-B_{d-2}(\V^{(d-2)}_b)\right]\displaybreak[3]\\
&=(-1)^{n+d}(-1)^{w-d}\binom{n-d+2}{n-w}\left[B_{d-2}(\V^{(d-2)}_a)-B_{d-2}(\V^{(d-2)}_b)\right]. \qedhere
\end{align*}
\end{proof}

The symmetry of the different weight distributions of the weight $d-2$ cosets (as well as the weight $2$ cosets in Section \ref{sec_wd2}) is a new interesting fact.

\begin{remark}\label{rem4:comp}
Considering the theorems and its proofs in this section, we see that the new version \eqref{eq1_BonGenTransf}, \eqref{eq1_Omega} of the Bonneau formula allows us to easily find simplifications for specific cases. The proof of the symmetry in Theorem \ref{th4_symmetry_d-2} is also relatively simple.
\end{remark}

The results of this section give the convenient tools for obtaining the weight distributions of the cosets of distinct MDS codes.

\section{The weight distribution of the weight 2 cosets of MDS codes of distance $d\ge5$}\label{sec_wd2}
In this section, using the results of Section \ref{sec_spec}, we consider the weight distribution of the weight 2 cosets of MDS codes in more detail.

\begin{theorem}\label{th5:W=2}
Let $\C$ be an $[n,k,d]_q$ MDS code of distance $d\ge5$. Let $\V^{(2)}$ be one of its cosets of weight $2$.  Assume that  the value of $B_{d-2}(\V^{(2)})$ is known. Then the number $B_w(\V^{(2)})$ of weight $w$ vectors in the coset $\V^{(2)}$  has the  form:
\begin{align}\label{eq5:Bon_w=2}
&B_w(\V^{(2)})=0 \text{ if } w\in\{0,1,\ldots,d-3\}\setminus\{2\},~B_2(\V^{(2)})=1,\displaybreak[3]\\
&B_w(\V^{(2)})=A_w(\C)-\Omega_w^{(0)}(n,d)+\Omega_w^{(2)}(n,d)+(-1)^{w-d}\binom{n-d+2}{n-w}B_{d-2}(\V^{(2)})\displaybreak[3]\notag\\
&\text{ if }w= d-1,d,\ldots,n.\notag
\end{align}
\end{theorem}

\begin{proof}
We use \eqref{eq4_Bon_w<=t} with $W=2$ and apply \eqref{eq4_Omega5}.
\end{proof}

The relation \eqref{eq5:Bon_w=2}  is the base of further researches in this section.
\begin{lemma}\label{lem5:Bd-2Sum}
Let $\C$ be an $[n,k,d]_q$ MDS code with $d\ge5$. Then the overall number $\B_{d-2}^{\Sigma}(\V^{(2)})$ of weight $d-2$ vectors in all the weight $2$ cosets of $\C$ is as follows:
\begin{align}\label{eq5:Bd-2Sum}
   \B_{d-2}^{\Sigma}(\V^{(2)})=(q-1)\binom{n}{2}\binom{n-2}{d-2}.
   \end{align}
\end{lemma}

\begin{proof}
 As $d\ge5$, each weight 2 coset $\V^{(2)}$ has a unique leader. For each weight $d$ codeword $\cb_d$ of $\C$, there are $\binom{d}{2}$ coset leaders $\vb_2$ such that $\cb_d+\vb_2=\vb_{d-2}$. Therefore, using \eqref{eq2_wd_MDS} and  the 1-st equality of \eqref{eq3_Riordan_ident2}, we have
  \begin{align*}
&  \B_{d-2}^{\Sigma}(\V^{(2)})=A_d(\C)\binom{d}{2}=(q-1)\binom{n}{d}\binom{d}{2}=(q-1)\binom{n}{2}\binom{n-2}{d-2}. \qedhere
  \end{align*}
\end{proof}
Lemma \ref{lem5:Bd-2Sum} allows us to obtain the weight distribution for the case when all the weight 2 cosets have the same distribution.
\begin{theorem}\label{th5:wd2_ident}
Let  $\C$ be an $[n,k,d]_q$ MDS code of distance $d\ge5$.    Assume that  all the $\binom{n}{2}(q-1)^2$ weight $2$ cosets $\V^{(2)}$ of $\C$ have the same weight distribution.
Then, the weight distribution  of any weight $2$ coset $\V^{(2)}$ is as follows:
  \begin{align}\label{e6_wd2_ident}
&B_w(\V^{(2)})=0 \text{ if } w\in\{0,1,\ldots,d-3\}\setminus\{2\},~B_2(\V^{(2)})=1,\displaybreak[3]\\
&B_{d-2}(\V^{(2)})=\frac{\B_{d-2}^{\Sigma}(\V^{(2)})}{\binom{n}{2}(q-1)^2}=\frac{1}{q-1}\binom{n-2}{d-2},\displaybreak[3]\label{eq5:wd2identBd-2}\\
&B_w(\V^{(2)})
=A_w(\C)-\Omega_w^{(0)}(n,d)+\Omega_w^{(2)}(n,d)
+(-1)^{w-d}\,\frac{1}{q-1}\binom{n-d+2}{n-w}\binom{n-2}{d-2}\displaybreak[3]\notag\\
&\text{if }w=d-1,d,\ldots,n.\notag
\end{align}
\end{theorem}

\begin{proof}
By \eqref{eq2:NWSigma}, the total number of the weight $2$ cosets $\V^{(2)}$ of $\C$ is $\binom{n}{2}(q-1)^2$ that together with  \eqref{eq5:Bd-2Sum} gives rise to \eqref{eq5:wd2identBd-2}. For the rest of relations we apply  \eqref{eq5:Bon_w=2}.
\end{proof}

We establish necessary condition for equality of the weight distributions.

\begin{theorem}\label{th5:necessar}
\textbf{(necessary condition for equality of weight distributions)} Let  $\C$ be an $[n,k,d]_q$ MDS code of distance $d\ge5$.  The necessary condition for equality of the weight distributions of all the weight $2$ cosets of $\C$ is as follows:
\begin{align}\label{eq6_necessar}
\frac{1}{q-1}\binom{n-2}{d-2} \text{ is an integer}.
\end{align}
\end{theorem}

\begin{proof}
All values of $B_w(\V^{(2)})$ must be integer. The assertion follows from~\eqref{eq5:wd2identBd-2}.
\end{proof}

\begin{lemma}\label{lem5:coprime}
    Let $n=q+1$, $d\ge5$. Let $q-1$ be co-prime with $d-2$. Then the necessary condition \eqref{eq6_necessar} holds.
\end{lemma}

\begin{proof}
    We have $n-2=q-1$ and $(d-2)|(q-1)(q-2)(q-3)\ldots(q-d+2)$. As $q-1$ is co-prime with $d-2$ we have also $(d-2)|(q-2)(q-3)\ldots(q-d+2)$.
\end{proof}

\begin{example}
    We consider the $[q+1,q-3,5]_q3$ GDRS code $\C$ with $q\not\equiv1\pmod3$.   By Lemma \ref{lem5:coprime}, the condition~\eqref{eq6_necessar} holds. In \cite[Theorem 19]{DMP_CosetsRScod4}, it is proved that all the weight 2 cosets of $\C$ have the same weight distribution. So, in this case,  the necessary condition is also sufficient.
\end{example}

Finally, we prove that similarly to the weight $d-2$ cosets,  different weight distributions of the weight 2 cosets are  symmetrical.
\begin{theorem}\label{th5:symmetr}
\textbf{(symmetry of different weight distributions)}
    Let $\C$ be an $[n,k,d]_q$ MDS code with $d\ge5$. Let $\V^{(2)}_a$ and $\V^{(2)}_b$ be two of its weight $2$ cosets with  different weight distributions. Then independently of the values of $B_{d-2}(\V^{(2)}_a)$ and $B_{d-2}(\V^{(2)}_b)$, there is the following symmetry of the weight distributions:
\begin{align*}
&(-1)^{n+d}B_w(\V^{(2)}_a)-B_{n+d-2-w}(\V^{(2)}_a)=(-1)^{n+d}B_w(\V^{(2)}_b)-B_{n+d-2-w}(\V^{(2)}_b),\displaybreak[3]\\
&w=d-1,d,\ldots,n.\notag
\end{align*}
\end{theorem}

\begin{proof}
The proof is similar to one of Theorem \ref{th4_symmetry_d-2}. Instead of \eqref{eq4_Bon_w=d-2}  we use~\eqref{eq5:Bon_w=2}.
\end{proof}

\section{Arcs in the projective plane $\PG(2,q)$ and the weight distribution of the cosets of MDS codes of distance $d=4$}\label{sec_d=3,4}
\subsection{Preliminaries on arcs in the projective plane}\label{subsec:prelim_arc}
We give some definitions and properties of arcs useful in this paper, see \cite{EtzStorm2016,Hirs_PGFF,HirsStor-2001} and the references therein.
\begin{definition}\label{def6:correspond}
  Let $\A$ be an $n$-arc in $\PG(2,q)$. An $[n,n-3,4]_qR$ MDS code $\C$ is said to be \emph{corresponding to the arc} $\A$ if the columns of its $3\times n$ parity check  matrix are the points of $\A$ in homogenous coordinates. Conversely, the arc $\A$ is \emph{corresponding to the code} $\C$.
\end{definition}
The code $\C$ of Definition \ref{def6:correspond} has covering radius $R=2$ (resp. $R=3$) if the arc $\A$ is complete (resp. incomplete).

In $\PG(2,q)$, a bisecant (resp. unisecant) of an arc is a line having two (resp. one) common points with the arc. Every $n$-arc has, in total, $\binom{n}{2}$ bisecants and $n(q+2-n)$ unisecants; there are $q+2-n$ unisecants in each point of the $n$-arc.

\begin{definition}\label{def6:ci}
 For an arc in $\PG(2,q)$, let $c_i$ be the number of the points off the arc lying on $i$ its bisecants.
\end{definition}

A complete (resp. incomplete) arc has $c_0=0$ (resp. $c_0>0$).

For $i>0$,
we note only non-zero values of~$c_i$; they allow us to obtain the weight distribution of the cosets of the MDS code corresponding to an arc.

The $q+1$ columns of $H_4$  \eqref{eq2_HDRS} (as well as the 1-st $q+1$ columns of $\widetilde{H_4}$ \eqref{eq2_HTRS}) can be treated as the points of a \emph{conic} $\K\subset\PG(2,q)$ in homogenous coordinates. The $q+2$ columns of $\widetilde{H_4}$  form a \emph{regular hyperoval} $\mathcal{H}\subset\PG(2,q)$.
The conic $\K$ and the hyperoval $\mathcal{H}$ have the following properties.

$\bullet$ Let $q$ be odd.
The conic $\K$ is a complete $(q+1)$-arc. Outside $\K$, there are $\N_\text{int}:=(q^2-q)/2$ internal points and $\N_\text{ext}:=(q^2+q)/2$  external points. Every internal and external point lies, respectively, on $\Bb_\text{int}:=(q+1)/2$ and $\Bb_\text{ext}:=(q-1)/2$ bisecants of~$\K$. Also, every external point lies on two unisecants of~$\K$.
So, if $q$ is odd, for the conic $\K$ we have $c_{i_1}=(q^2+q)/2,~i_1=(q-1)/2,~ c_{i_2}=(q^2-q)/2,~i_2=(q+1)/2$.
The matrix $H_4$  is a $3\times(q+1)$ parity check matrix of the $[q+1,q-2,4]_q2$ GDRS code.

$\bullet$ Let $q$ be even.
The conic $\K$ is an incomplete $(q+1)$-arc. Outside $\K$, there are $\N_\text{ev}:=q^2-1$ points all of which lie on $\Bb_\text{ev}:=q/2$ bisecants of $\K$ and on one its unisecant. The so-called nucleus $\mathcal{O}=(0,1,0)^{tr}$ does not lie on any bisecant of $\K$; it is the intersection of all $q+1$ unisecants of~$\K$. So, if $q$ is even, for the conic $\K$ we have $c_0=1,~c_{q/2}=q^2-1$; also, $H_4$ (as well as the 1-st $q+1$ columns of $\widetilde{H_4}$) is a parity check matrix of the $[q+1,q-2,4]_q3$ GDRS code.
The hyperoval $\mathcal{H}$ is a complete $(q+2)$-arc. Every point outside $\mathcal{H}$ lies on $(q+2)/2$ bisecants of $\mathcal{H}$. So, for the regular hyperoval $\mathcal{H}$ we have $c_{(q+2)/2}=q^2-1$; also, $\widetilde{H_4}$ is a parity check matrix of the $[q+2,q-1,4]_q2$ GTRS code.

Note that for any $q \ge 16$, also non-regular hyperovals exist \cite[Theorem 8.37]{Hirs_PGFF}. They have the same parameters as the regular ones but their structures are different. The points of non-regular hyperovals can form parity check matrices of MDS codes which formally are not GTRS, GDRS, or GRS. In this paper, we do not consider the non-regular hyperovals; they will be a subject of our future investigations.

The hyperovals are the unique (up to parameters) complete arcs with only one non-zero value of $c_i$ \cite[Theorem 9.14]{Hirs_PGFF}. All other complete arcs, for even and odd $q$, have at least two non-zero $c_i$.

\looseness=-1 There are several examples (including infinite families with increasing $q$) of complete $n$-arcs in $\PG(2,q)$ with $n<q+1$, see \cite{Hirs_PGFF,HirsStor-2001,HirsThas-2015} and the references therein. These arcs are not conics or hyperovals; the corresponding codes are not GDRS or GTRS.   In the literature, to the best of our knowledge, the values of $c_i$ for infinite families of complete $n$-arcs with $n<q+1$ are not described. But, for some sporadic complete $n$-arcs in $\PG(2,q)$ these values are given, see e.g. in \cite[Section~9]{Hirs_PGFF}:
\begin{align}\label{eq6:specific}
&n=6,q=7,~(c_1,c_2,c_3)=(18,27,6);~n=6,q=8,~(c_1,c_2,c_3)=(36,24,7);\displaybreak[3]\\
&n=6,q=9,~(c_1,c_2,c_3)=(60,15,10);~n=7,q=11,~(c_1,c_2,c_3)=(63,42,21).\notag
\end{align}

\subsection{New results}
In this subsection, for GDRS and GTRS codes with distance $d=4$ and their shortenings, we obtain the weight distribution of the cosets, see Theorems \ref{th6:n=q+1}, \ref{th6:n=q+2}, \ref{th6:n=q_codes}, and \ref{th6:n=q-1_codes}. We use
the results of Section 4 and properties of the arcs corresponding to the codes in accordance to Definition \ref{def6:correspond}. For GDRS and GTRS codes, the needed arc properties are taken from Subsection \ref{subsec:prelim_arc} whereas, for the shortened codes, we obtain these properties in Propositions \ref{prop6:n=q_arcs} and \ref{prop6:n=q-1_arcs}.

\looseness=-1 Theorem \ref{th6:general} establishes connections between properties of $n$-arcs in $\PG(2,q)$ and the weight distributions of the cosets of the corresponding $[n,n-3,4]_qR$ MDS codes.

\begin{theorem}\label{th6:general}
Let the values $c_i$ be as in Definition \ref{def6:ci}. Let $\A$ be an $n$-arc in $\PG(2,q)$ with $c_{i_j}\ne0$, $i_j>0$, $j=1,2,\ldots,m$, $m\ge1$; also, both the cases $c_0=0$ and $c_0\ne0$ are possible. Let $\C$ be an $[n,n-3,4]_qR$ MDS code corresponding to $\A$ in accordance with Definition \ref{def6:correspond}. Then the following holds.
\begin{description}
  \item[(i)] The code $\C$ has $n(q-1)$ cosets $\V^{(1)}$ of weight $1$ with the weight distribution
   \begin{align}\label{eq6:wd1wc}
&B_0(\V^{(1)})=B_2(\V^{(1)})=0,~B_1(\V^{(1)})=1,~B_{3}(\V^{(1)})=\binom{n-1}{3},\displaybreak[3]\\
&B_w(\V^{(1)})=A_{w}(\C)-\Omega_w^{(0)}(n,4)+\Omega_w^{(1)}(n,4) \text{ if } w=4,5,\ldots,n.\notag
\end{align}

  \item[(ii)]  The weight $2$ cosets $\V^{(2)}$ of $\C$ can be partitioned into $m$ classes so that the $j$-th class consists of $(q-1)c_{i_j}$ cosets with the weight distribution
      \begin{align}\label{eq6:Bon_w=d-2}
&B_0(\V^{(2)})=B_1(\V^{(2)})=0,~B_2(\V^{(2)})=c_{i_j},\displaybreak[3]\\
&B_w(\V^{(2)})=A_w(\C)-\Omega_w^{(0)}(n,4)+(-1)^{w}\binom{n-2}{n-w}c_{i_j}\text{ if }w=3,4,\ldots,n.\notag
\end{align}
       Moreover, if $m>1$ then the weight distributions of the cosets of distinct classes are symmetrical in the sense of Theorem \ref{th4_symmetry_d-2}.

  \item[(iii)]  For $c_0=0$, the arc $\A$ is complete, $\C$ is an $[n,n-3,4]_q2$ code, we have no weight $3$ cosets.
  For $c_0\ne0$, the arc $\A$ is incomplete, $\C$ is an $[n,n-3,4]_q3$ code having $(q-1)c_0$ cosets of weight $3$ with the weight distribution
  \begin{align}\label{eq6:Bon_w=d-1}
&B_0(\V^{(3)})=B_1(\V^{(3)})=B_2(\V^{(3)})=0,~B_{3}(\V^{(3)})=\binom{n}{3},\displaybreak[3]\\
&B_w(\V^{(3)})=A_w(\C)-\Omega_w^{(0)}(n,4)\text{ if } w=4,5,\ldots,n.\notag
\end{align}
  \end{description}
\end{theorem}

\begin{proof}
 Every point $P$ of $\PG(2,q)$ gives rise to $q-1$ syndromes of the cosets; they can be obtained multiplying $P$ (in homogeneous coordinates) by elements of~$\F_q^*$.
\begin{description}
  \item[(i)]
 The points of $\A$ (i.e. the $n$ columns of the parity check matrix of $\C$) generate the syndromes of the $n(q-1)$ cosets of weight 1; \eqref{eq6:wd1wc} follows from~\eqref{eq4_wd1wc}.
  \item[(ii)] Obviously, the $c_{i_j}$ points off $\A$ lying on its $i_j$ bisecants give the syndromes of the $(q-1)c_{i_j}$ cosets of weight 2. The distribution \eqref{eq6:Bon_w=d-2} follows from~\eqref{eq4_Bon_w=d-2}.
  \item[(iii)] The points off the incomplete arc $\A$, that do not lie on any bisecant of $\A$, give the syndromes of the $(q-1)c_0$ cosets of weight~$3$;
 \eqref{eq6:Bon_w=d-1} follows from~\eqref{eq4_Bon_w=d-1}.\qedhere
\end{description}
\end{proof}

Using the relations of Section \ref{sec_spec}, Theorem \ref{th6:general}  expresses the weight distributions of the cosets directly by means of the values of $c_i$ characterizing arcs.

In Theorems \ref{th6:n=q+1} and \ref{th6:n=q+2}, we obtain the weight distributions of all the cosets of the MDS codes corresponding to the conic and the regular hyperoval in $\PG(2,q)$.

\begin{theorem}\label{th6:n=q+1}
Let $\C$ be the $[q+1,q-2,4]_qR$ GDRS code  with the parity check matrix $H_4$ \eqref{eq2_HDRS} corresponding to the conic $\K$ in the projective plane $\PG(2,q)$.
\begin{description}
  \item[(i)] Let $q$ be odd. Then the covering radius is $R=2$ and the weight distribution of all the  cosets of $\C$ is as in Theorem \ref{th6:general} for the values  $m=2,~c_{i_1}=(q^2+q)/2,~i_1=(q-1)/2,~ c_{i_2}=(q^2-q)/2,~i_2=(q+1)/2,~c_0=0$.

  \item[(ii)] Let $q$ be even. Then $R=3$ and the weight distribution of all the  cosets of $\C$ is as in Theorem \ref{th6:general} for
  $m=1,~c_{i_1}=q^2-1,~i_1=q/2,~c_0=1$.
\end{description}
\end{theorem}

\begin{proof}
The assertions follow from Theorem \ref{th6:general} and the data in Subsection~\ref{subsec:prelim_arc}.
\end{proof}

\begin{theorem}\label{th6:n=q+2}
Let $q$ be even. Let $\C$ be the $[q+2,q-1,4]_q2$ GTRS code with the parity check matrix $\widetilde{H_4}$ \eqref{eq2_HTRS} corresponding to the regular hyperoval $\mathcal{H}$ in the plane $\PG(2,q)$.
Then the weight distribution of all the  cosets of $\C$ is as in Theorem~\emph{\ref{th6:general}} for the values
  $m=1,~c_{i_1}=q^2-1,~i_1=(q+2)/2,~c_0=0$.
\end{theorem}

\begin{proof}
The assertions follow from Theorem \ref{th6:general} and the data in Subsection~\ref{subsec:prelim_arc}.
\end{proof}

\begin{remark}
  The weight distribution of the cosets of the $[n,n-3,4]_q2$ MDS codes corresponding to the complete arcs from \eqref{eq6:specific} can be obtained directly from Theorem \ref{th6:general} on the base of the values of $c_i$ given in \eqref{eq6:specific}.
\end{remark}

In Propositions \ref{prop6:n=q_arcs} and \ref{prop6:n=q-1_arcs}, for shortened conics, we obtain values of $c_i$ not known in the literature. Using these $c_i$ and Theorem \ref{th6:general}, we obtain the weight distributions of all the cosets of the corresponding MDS codes, see Theorems \ref{th6:n=q_codes} and \ref{th6:n=q-1_codes}.

\begin{proposition}\label{prop6:n=q_arcs}
 Let $q\ge5$. In $\PG(2,q)$, let $P$ be a point of the conic $\K$. Let $\K^*=\K\setminus\{P\}$ be the ``shortened'' conic. For~$\K^*$, the non-zero values of $c_{i_j}$ corresponding to Definition \ref{def6:ci} and Theorem \ref{th6:general} are as follows:
 \begin{description}
   \item[(i)]  Let $q$ be odd. Then $m=2$, $c_{i_1}=(q^2+q)/2,~i_1=(q-1)/2$; $c_{i_2}=(q^2-q)/2$, $i_2=(q-3)/2$; $c_0=1$.

   \item[(ii)]  Let $q$ be even. Then $m=2$, $c_{i_1}=q-1,~i_1=q/2$; $c_{i_2}=q^2-q$, $i_2=(q-2)/2$; $c_0=2$.
  \end{description}
 \end{proposition}

\begin{proof}
 \begin{description}
   \item[(i)]  The point $P$ does not lie on any bisecant of~$\K^*$.
The unisecant  of $\K$ in $P$ contains $q$ external points \cite[Table 8.1]{Hirs_PGFF} lying on $\Bb_\text{ext}$ bisecants of $\K^*$ as well as for $\K$. Each of the remaining $\N_\text{ext} - q$ external points loses one bisecant of $\K$ after removing $P$, i.e. it lies on $\Bb_\text{ext}-1$ bisecants of $\K^*$. Each internal point also loses one bisecant  of $\K$, i.e. it lies on $\Bb_\text{int}-1$ bisecants of~$\K^*$.
Based on the above, the assertion follows.
 Note that $i_2\ge1$ as $q\ge5$.

   \item[(ii)]  The nucleus $\mathcal{O}$ and the point $P$ do not lie on any bisecant of $\K^*$; this gives $c_0=2$.  The unisecant $\U$ of $\K$ in~$P$ contains $\mathcal{O}$ and $q-1$ points lying on $\Bb_\text{ev}$ bisecants of $\K^*$ as well as for~$\K$.  Each point of $\PG(2,q)\setminus(\K^*\cup\U)$ loses one bisecant  of $\K$ after removing $P$, i.e. it lies on $\Bb_\text{ev}-1$ bisecants of $\K^*$. Based on the above, the assertion follows. \qedhere
 \end{description}
\end{proof}

\begin{theorem}\label{th6:n=q_codes}
Let $\C$ be a $[q,q-3,4]_qR$ GRS code  with the parity check matrix  obtained from the matrix $H_4$ \eqref{eq2_HDRS} by removing any one column. Let $q\ge5$. Then the covering radius is $R=3$, the code $\C$ corresponds to the shortened conic $\K^*$ from Proposition \ref{prop6:n=q_arcs}, and the weight distribution of all the  cosets of $\C$ is as in Theorem~\ref{th6:general} for the values $c_{i_j}$ taken from Proposition \ref{prop6:n=q_arcs}.
\end{theorem}

\begin{proof}
The assertions follow from Theorem \ref{th6:general} and Proposition \ref{prop6:n=q_arcs}.
\end{proof}

  If the parity check matrix of the code $\C$ of Theorem \ref{th6:n=q_codes} is $H_4$ \eqref{eq2_HDRS} without the column $(0,0,v_{q+1})^{tr}$ then $\C$ is a \emph{single-extended} $[q,q-3,4]_q3$ GRS code.

\begin{proposition}\label{prop6:n=q-1_arcs}
Let $q\ge7$. In $\PG(2,q)$, let $P_1,P_2\in\K$ be points of the conic~$\K$. Let $\K^{**}=\K\setminus\{P_1,P_2\}$ be the ``double shortened'' conic. Then, for $\K^{**}$, the non-zero values of $c_{i_j}$ corresponding to Definition \ref{def6:ci} and Theorem \ref{th6:general} are as follows:
 \begin{description}
   \item[(i)]  Let $q$ be odd. Then $m=3$, $c_{i_1}=(q+1)/2,~i_1=(q-1)/2$; $c_{i_2}=(q-1)(q+4)/2$, $i_2=(q-3)/2$; $c_{i_3}=(q-1)(q-3)/2$, $i_3=(q-5)/2$; $c_0=2$.

   \item[(ii)]  Let $q$ be even. Then $m=2$, $c_{i_1}=3(q-1),~i_1=(q-2)/2$; $c_{i_2}=(q-1)(q-2)$, $i_2=(q-4)/2$; $c_0=3$.
  \end{description}
 \end{proposition}

 \begin{proof}
 \begin{description}
   \item[(i)] The points $P_1,P_2$ provide $c_0=2$.
Each unisecant $\U_1$ and $\U_2$ of $\K$ in $P_1$ and $P_2$, respectively, contains $q$ external points \cite[Table 8.1]{Hirs_PGFF};  $2(q-1)$ of these points lie on one unisecant and lose one bisecant of $\K$, i.e.\ each of them lies on $\Bb_\text{ext}-1$ bisecants of~$\K^{**}$.   The point $\U_1\cap\U_2$ lies on $\Bb_\text{ext}$ bisecants of~$\K^{**}$.

The bisecant $\overline{P_1,P_2}$ of $\K$ through $P_1$ and $P_2$ contains $\frac{1}{2}(q-1)$ external points as well as internal ones \cite[Table 8.1]{Hirs_PGFF}. Each of these points loses one bisecant of $\K$, i.e. they lie on $\Bb_\text{ext}-1$ and $\Bb_\text{int}-1$ bisecants of $\K^{**}$, respectively.
 Each of the remaining $\N_\text{ext} - 2(q-1)-1-\frac{1}{2}(q-1)$ external points and $\N_\text{int}-\frac{1}{2}(q-1)$ internal ones loses two bisecants, i.e. they lie on $\Bb_\text{ext}-2$ and $\Bb_\text{int}-2$ bisecants of $\K^{**}$.
Based on the above, the assertion follows.
 Note that $i_3\ge1$ as $q\ge7$.

   \item[(ii)]  The nucleus $\mathcal{O}$ and the points $P_1,P_2$  give rise to $c_0=3$.  The unisecants $\U_1$ and $\U_2$ of $\K$ in $P_1$ and $P_2$ contain $\mathcal{O}$  and $2(q-1)$ points every of which loses one bisecant and lies on $\Bb_\text{ev}-1$ bisecants of~$\K^{**}$.  The same properties hold for $q-1$ points of the bisecant $\overline{P_1,P_2}$. The remaining $\N_\text{ev}-2(q-1)-(q-1)$ points  lose  two bisecants after removing $\{P_1,P_2\}$, i.e.\ each of them lies  on $\Bb_\text{ev}-2$ bisecants of $\K^{**}$. Based on the above, the assertion follows.\qedhere
 \end{description}
\end{proof}

\begin{theorem}\label{th6:n=q-1_codes}
Let $\C$ be the $[q-1,q-4,4]_qR$ GRS code with the parity check matrix  obtained from the matrix $H_4$ \eqref{eq2_HDRS} by removing any two columns. Let $q\ge7$.
Then the covering radius is $R=3$, the code $\C$ corresponds to the double shortened conic $\K^{**}$ from Proposition \ref{prop6:n=q-1_arcs}, and the weight distribution of all the  cosets of $\C$ is as in Theorem~\ref{th6:general} for the values $c_{i_j}$ taken from Proposition \ref{prop6:n=q-1_arcs}.
\end{theorem}

\begin{proof}
The assertions follow from Theorem \ref{th6:general} and Proposition \ref{prop6:n=q-1_arcs}.
\end{proof}

  If the parity check matrix of the code $\C$ of Theorem  \ref{th6:n=q-1_codes} is  $H_4$ \eqref{eq2_HDRS} without $(v_{q},0,0)^{tr}$ and $(0,0,v_{q+1})^{tr}$ then $\C$ is a \emph{non-extended} $[q-1,q-4,4]_qR$ GRS code.

The codes considered in this section are GDRS, GTRS, single-extended GRS, and non-extended GRS ones. They are interesting for coding theory.

\section{MDS codes as multiple coverings of deep holes and multiple saturating sets in $\PG(N,q)$}\label{sec_multip}

\subsection{Preliminaries on multiple coverings}\label{subsec:multip_prelim}
We give some definitions and propositions of \cite{BDGMP_MultCov,BDGMP_MultCovFurth,BDMP_TwistCub,CHLL_CovCodBook}, see also the references therein.

\begin{definition}\label{def7_multcov}
\begin{description}
  \item[(i)] An $[n,k,d]_{q}R$ code $\C$ is said to be an \emph{$(R,\mu )$
multiple covering of the farthest-off points} ($(R,\mu )$-MCF
code for short) if for all $\x\in \F_{q}^{n}$ with $ d(\x,\C)=R$ we have $f_R(\x,\C)\ge\mu$. The parameter $\mu$ is called the multiplicity of covering.
  \item[(ii)] An $[n,k,d(\C)]_{q}R$ code $\C$ is said to be an \emph{$(R,\mu )$
almost perfect} multiple covering of the farthest-off points
($(R,\mu )$-APMCF code for short) if for all $\x\in \F_{q}^{n}$ with $ d(\x,\C)=R$ we have $f_R(\x,\C)=\mu$. If,
in addition, $d(\C)\geq2R$, then the code is called
\emph{$(R,\mu )$  perfect} multiple covering of the farthest-off
points $((R,\mu )$-PMCF code for short).
\end{description}
\end{definition}

In the literature, MCF codes are also called \emph{multiple coverings of deep holes}.

The covering quality of an $[n,k,d(\C)]_{q}R$ $(R,\mu )$-MCF code $\C$ is characterized by its \emph{$\mu$-density} $\gamma _{\mu }(\C,R,q)\ge1$ that is the average value of $f_R(\x,\C)$ divided by $\mu$ where the average is calculated over all $\x\in \F_{q}^{n}$ with $ d(\x,\C)=R$. For APMCF and PMCF codes we have $\gamma _{\mu }(\C,R,q)=1$. Also,
\begin{align}
&\gamma _{\mu }(\C,2,q)=\frac{{\binom{n}{2}}(q-1)^{2}}{\mu (q^{n-k}-1-n(q-1))}\text{ if }d(\C)>
3.  \label{eq7_mudens}
\end{align}
From the covering problem point of view, the best codes are those with small $\mu$-density.

If the $\mu$-density $\gamma_\mu(\C,R,q)$ tends to 1 when $q$ tends to infinity we have an \emph{asymptotically optimal collection of MCF codes} or, in another words, an \emph{asymptotically optimal multiple covering}.

\begin{definition}
Let $S$ be a subset of points of
$\PG(N,q)$. Then $S$ is said to be $(\rho ,\mu )$-saturating if the following holds:
$S$ generates $\PG(N,q)$; there exists a point $Q$ in $\PG(N,q)$ which does not belong to
any subspace of dimension $\rho -1$ generated by the points of $S$;
and, finally, every point $Q\in\PG(N,q)$ not belonging to any subspace of
dimension $\rho -1$ generated by the points of $S$, is such that the number
of subspaces of dimension $\rho $ generated by the points of $S$ and
containing $Q$, counted with multiplicity, is at least $\mu $. The
multiplicity $m_{T}$ of a subspace $T$ is computed as
the number of distinct sets of $\rho +1$
independent points contained in $T\cap S$. If any $\rho +1$ points of $S$ are linearly independent, then
$m_{T}={\binom{{\#(T\cap S)}}{{\rho +1}}.}$
\end{definition}

\begin{definition}
\label{def OS}Let $S$ be a $(\rho ,\mu)$-saturating set in
$ \PG(N,q). $ The set $S$ is called \emph{optimal }$(\rho
,\mu)$-{\em saturating set }$((\rho ,\mu )$-OS set for
short) if every point $Q$ in $\PG(N,q)$ not belonging to
any subspace of dimension $\rho-1$ generated by the points of
$S$, is such that the number of subspaces of dimension $\rho $
generated by the points of $S$ and containing $Q$, counted with multiplicity, is
exactly $\mu $.
\end{definition}

\begin{proposition}\label{prop7}
    Let $\C$ be an $[n,k,d=n-k+1]_{q}R$ MDS code with an $(n-k)\times n$ parity check matrix $H$. Let $\C$ be an $(R,\mu)$-MCF \emph{(}resp. $(R,\mu)$-APMCF\emph{)} code. Let $S$ be the $n$-set of points in $\PG(n-k-1,q)$ such that its points in homogeneous coordinates are columns of $H$. Then $S$ is an $(R-1,\mu)$-saturating \emph{(}resp. $(R-1,\mu)$-OS\emph{)} set corresponding to $\C$.
\end{proposition}

Proposition \ref{prop7} allows us to consider linear MDS $(R,\mu )$-MCF (resp. $(R,\mu)$-APMCF) codes as $(R-1,\mu )$-saturating (resp. $(R-1,\mu)$-OS) sets and vice versa. Thus, there is the one-to-one correspondence between MCF codes and multiple saturating sets in the projective spaces $\PG(N,q)$.

\subsection{New results}\label{subsec:multip_newres}
In this subsection, using Sections \ref{sec_spec}  and \ref{sec_d=3,4}, we obtain a few new results.

 \begin{lemma}\label{lem7}
    Let $\C$ be an $[n,k,d(\C)]_{q}R$ code and $\V^{(R)}$ be one of its weight $R$ cosets.
    \begin{description}
      \item[(i)] Any vector $\x$ of the coset $\V^{(R)}$ is a farthest-off vector (deep hole) for $\C$. The number of codewords on distance $R$ from $\x$ is equal to $B_R(\V^{(R)})$.
      \item[(ii)] If for all cosets $\V^{(R)}$ we have $B_R(\V^{(R)})\ge\mu$ then $\C$ is an $(R,\mu )$-MCF code.
      \item[(iii)] If for all cosets $\V^{(R)}$ we have $B_R(\V^{(R)})=\mu$ then $\C$ is an $(R,\mu )$-APMCF  code; moreover, if,
in addition, $d(\C)\geq2R$, then $\C$ is an $(R,\mu )$-PMCF code.
    \end{description}
 \end{lemma}

\begin{proof}
For any vector $\x\in\V^{(R)}$ we have $d(\x,\C)=R$, $f_R(\x,\C)=B_R(\V^{(R)})$, see Section \ref{subsec_cosets}.
This proves the case (i). The cases (ii), (iii) follow from Definition~\ref{def7_multcov}.
\end{proof}

Lemma \ref{lem7} establishes natural connections between the weight distributions of the weight $R$ cosets of $[n,k,d]_{q}R$ codes and the parameters of MCF ones.

The following Theorem \ref{th7_R=d-1} shows that there is a wide class of MDS codes for which the weight distributions of the weight $R$ cosets can be obtained on the base of the results of Sections  \ref{sec_spec}  and \ref{sec_d=3,4}.

\begin{theorem}\label{th7_R=d-1}
 Let $\C$ be an  $[n,k,d]_qR$ MDS code of distance $d\ge3$, covering radius $R=d-2$, and with a parity check matrix $H_\C$. In particular, $\C$ could be the GDRS code with the parity check matrix $H_d$~\eqref{eq2_HDRS}, or the GTRS code with the parity check matrix $\widetilde{H_4}$ \eqref{eq2_HTRS}, or one of the codes of \cite{BartGiulPlat}. Assume that we remove $\Delta$ columns from $H_\C$, $1\le\Delta\le n-d$. Then we obtain a parity check matrix $H_{\C_\Delta}$ of an $[n-\Delta,k-\Delta,d]_{q}R_\Delta$ MDS code $\C_\Delta$ of covering radius $R_\Delta=d-1$. Moreover, all the weight $R_\Delta$ cosets $\V^{(R_\Delta)}$ of $\C_\Delta$ have the same weight distribution $B_w(\V^{(R_\Delta)})=B_w(\V^{(d-1)})$ of the form \eqref{eq4_Bon_w=d-1} with $B_{R_\Delta}(\V^{(R_\Delta)})=\binom{n}{d-1}$.
\end{theorem}

\begin{proof}
The columns, removed from $H_\C$, are not a linear combination of $\le d-2$ columns of $H_\C$, otherwise the distance of $\C$ would be $<d$. But the removed columns can be obtained as a linear combination of $d-1$ columns of $H_{\C_\Delta}$ since any $d$ columns of $H_\C$ are linearly dependent. So,  $R_\Delta=d-1$ and we may use Theorem \ref{th4_W=d-1}.
\end{proof}

\begin{theorem}\label{th7_R=d-1_2}
Let $\C$ be an $[n,k,d]_qR$ MDS code of distance $d\ge3$ and covering radius $R=d-1$; in particular, it could be the code $\C_\Delta$ of Theorem \ref{th7_R=d-1}. Then there are exactly $\binom{n}{d-1}$ codewords at distance $R$ from every farthest-off vector (deep hole), i.e.  $\C$ is a $(d-1,\mu)$-APMCF code with $\mu=\binom{n}{d-1}$.
\end{theorem}

\begin{proof}
    By Theorem \ref{th4_W=d-1}, all the weight $R$ cosets $\V^{(R)}$ of $\C$ have the same weight distribution with $B_{R}(\V^{(R)})=\binom{n}{d-1}$. Now we use  Lemma~\ref{lem7}(iii).
\end{proof}

Thus, MDS codes of covering radius $R=d-1$ are almost perfect MCF codes. By Theorem \ref{th7_R=d-1}, any shortening of an MDS code gives a code with $R=d-1$, i.e. these codes form a wide class and they can be constructed for any $R$. This is an important result as in the literature, see e.g. \cite{BDGMP_MultCov,BDGMP_MultCovFurth,BDMP_TwistCub,CHLL_CovCodBook},  almost perfect  MCF codes (including MDS ones) are described only for $R\le3$.

Note also that perfect  MCF codes (PCMF) have rarely appeared in the literature. In \cite[Proposition 4]{BDGMP_MultCovFurth}, it is shown that the code corresponding to the hyperoval is PCMF.  To the authors' knowledge, it is the only known MDS code with such a property. The following proposition proves that this code is unique.

\begin{proposition}\label{prop7_PMCF}
Let $q$ be even. Let $\C$ be the  $[q+2,q-1,4]_q2$ GTRS code with the parity check matrix  $\widetilde{H_4}$ \eqref{eq2_HTRS} corresponding to the regular hyperoval $\mathcal{H}$ in $\PG(2,q)$.
Then all the farthest-off vectors (deep holes) are at distance $2$ from $(q+2)/2$ codewords, i.e.  $\C$ is a $(2,\mu)$-PMCF code with $\mu=(q+2)/2$. Moreover, $\C$ is the unique (up to parameters) $[n,n-3,4]_q2$ code such that it can be viewed as a $(2,\mu)$-PMCF one.
\end{proposition}

\begin{proof}
By Theorems \ref{th6:general} and \ref{th6:n=q+2}, for all the weight 2 cosets of $\C$  we have $B_{2}(\V^{(2)})=(q+2)/2$. Then we use Lemma \ref{lem7}(iii).
The uniqueness of $\C$ follows from \cite[Theorem~9.14]{Hirs_PGFF} where it is proved that if all the points off a complete $n$-arc in $\PG(2,q)$ lie on $i$ its bisecants then $q$ is even, $i=(q+2)/2$, and $n=q+2$.
\end{proof}

In \cite{BDMP_TwistCub}, an asymptotical optimal collection of MCF codes with distance $d=5$ is obtained. The following proposition gives a new example with $d=4$.
\begin{proposition}\label{prop7_MCF}
Let $q$ be odd. Let $\C$ be the  $[q+1,q-2,4]_q2$ GDRS code with the parity check matrix  $H_4$ \eqref{eq2_HDRS} corresponding to the conic $\K$ in $\PG(2,q)$.
Then $q^{q-2}(q-1)(q^2+q)/2$ and $q^{q-2}(q-1)(q^2-q)/2$ deep holes are at distance~$2$ from $(q-1)/2$ and $(q+1)/2$ codewords, respectively. The code $\C$ is a $(2,\mu)$-MCF with $\mu=\frac{1}{2}(q-1)$ and $\mu$-density $\gamma _{\mu }(\C,2,q)=1+\frac{1}{q}$ which tends to~$1$ when $q$ tends to infinity.
          Thereby, we have an asymptotical optimal collection of MCF codes.
\end{proposition}

\begin{proof}
Every point $P$ off $\K$ gives rise to $q-1$ syndromes of weight 2 cosets of $\C$; they can be generated by multiplying $P$ (in homogeneous coordinates) by elements of~$\F_q^*$. Each coset contains $q^{q-2}$ vectors. By Theorem \ref{th6:n=q+1}(i), there are $(q^2+q)/2$ and $(q^2-q)/2$ points off $\K$  on $(q-1)/2$ and $(q+1)/2$ its bisecants, respectively. Thus, by Theorems \ref{th6:general} and  \ref{th6:n=q+1}(i), for the weight 2 cosets of $\C$ we have  $B_{2}(\V^{(2)})\ge(q-1)/2$.  Then we use Lemma \ref{lem7}(ii). The $\mu$-density is calculated by~\eqref{eq7_mudens}.
\end{proof}

The farthest-off vectors (deep holes) of a code play an important role in decoding. Their classification is interesting. In \cite{KaipaIEEE2017,ZDK_DHRIEEE2019}, see also the references therein, the classification is based on the consideration of the deep holes as polynomials.

On the other hand, the number of codewords at the distance $R$\/ from a deep hole and  the weight distribution of the deep holes, see Theorems~\ref{th7_R=d-1},~\ref{th7_R=d-1_2} and  Propositions \ref{prop7_PMCF}, \ref{prop7_MCF}, can also be treated as a classification. In addition, it is useful to consider  the total number (or its estimate) of the deep holes. In Propositions \ref{prop7_PMCF}, \ref{prop7_MCF} this number is given. For the codes of Theorem \ref{th7_R=d-1} we prove the following.

\begin{theorem}\label{th7_NN}
Let $\C$ and $\C_\Delta$ be, respectively, the $[n,k,d]_qR$ and $[n-\Delta,k-\Delta,d]_{q}R_\Delta$ MDS codes of Theorem \ref{th7_R=d-1} with distance $d\ge3$ and covering radii $R=d-2$, $R_\Delta=d-1$. Then the total number $\mathbb{N}_\Sigma^{(d-1)}(\C_\Delta)$ of the weight $d-1$ cosets of the code $\C_\Delta$ is $\mathbb{N}_\Sigma^{(d-1)}(\C_\Delta)\ge(q-1)\Delta$ and the total number of the deep holes of $\C_\Delta$ is equal to $q^{k-\Delta}\mathbb{N}_\Sigma^{(d-1)}(\C_\Delta)$. Also, regarding $\mathbb{N}_\Sigma^{(d-1)}(\C_\Delta)$ the following holds:
\begin{description}
  \item[(i)]
 Let $\C$ be the  $[q+1,q-d+2,d]_q2$ GDRS code with the parity check matrix  $H_d$~\eqref{eq2_HDRS}, $d\in\{3,4\}$.
Then $\mathbb{N}_\Sigma^{(d-1)}(\C_\Delta)=(q-1)\Delta$ if $\Delta\le q+1-n_d(q)$ where
\begin{align}\label{eq7:ndq}
 n_3(q)= 3 ~\text{for any }q ,~n_4(q)= (q+5)/2 ~\text{for odd }q.
 \end{align}

  \item[(ii)]
Let $q$ be even. Let $\C$ be the  $[q+2,q-1,4]_q2$ code with the parity check matrix  $\widetilde{H_4}$ \eqref{eq2_HTRS}.
Let $\Delta\le q+2-(1+(q+4)/2)$ and let the column $(0,v_2,0)^{tr}$ corresponding to the hyperoval nucleus be removed. Then $\mathbb{N}_\Sigma^{(d-1)}(\C_\Delta)=\linebreak (q-1)\Delta$.
\end{description}
\end{theorem}

\begin{proof}
Similarly to the proof of Theorem \ref{th7_R=d-1}, we see that each column removed from $H_C$ can be obtained as a linear combination of $d-1$ columns of $H_{\C_\Delta}$. Also, each removed column gives rise to $q-1$ syndromes of the weight $d-1$ cosets. Therefore, $\mathbb{N}_\Sigma^{(d-1)}(\C_\Delta)\ge(q-1)\Delta$. Moreover, in the cases (i) and (ii), we have no other weight $d-1$ cosets. Each coset contains $q^{k-\Delta}$ vectors.
\begin{description}
  \item[(i)]
 For $d=3$ we consider the shortened $[q+1-\Delta,q-1-\Delta,3]_q2$ Hamming code. There are $q^2$ its cosets, one of which is the code itself and  $(q+1-\Delta)(q-1)$ ones have weight 1. The remaining $(q-1)\Delta$ cosets have weight $2=d-1$.

 Let $q$ be odd, $d=4$. We use the approaches of Section~\ref{sec_d=3,4}. Let $\K_\Delta$ be the arc obtained by removing $\Delta$ points from the conic~$\K$. We consider the $[q+1-\Delta,q-2-\Delta,4]_q3$ code $\C_\Delta$ corresponding to $\K_\Delta$. Each point off $\K$ lies on at least $\mathbb{B}_\text{ext}=(q-1)/2$ its bisecants. By \eqref{eq7:ndq}, $\mathbb{B}_\text{ext}-\Delta\ge1$. So, only the points removed from $\K$ generate the syndromes of the weight~3 cosets.

 \item[(ii)]
 We do similarly to the case (i). We consider the conic $\K$ obtained removing the nucleus, the arc $\K_{\Delta-1}$,  and the code $\C_\Delta$. Every point off $\K$ lies on  $\mathbb{B}_\text{ev}=q/2$ its bisecants. We have $\mathbb{B}_\text{ev}-(\Delta-1)\ge1$. So, only the points removed from $\K$ together with the nucleus generate the syndromes of the weight 3 cosets.
 \qedhere
\end{description}
\end{proof}

In conclusion of this section, we use the one-to-one correspondence between MCF codes and multiple saturating sets in $\PG(N,q)$.

\begin{theorem}\label{th7_R=d-1_2ccc}
Let $\C$ be the $[n,k,d]_qR$ MDS code of Theorem \ref{th7_R=d-1_2} that is a $(d-1,\mu)$-APMCF code with $d\ge3$, $R=d-1$, and $\mu=\binom{n}{d-1}$.  Let $S$ be the $n$-set of points in $\PG(n-k-1,q)$ such that its points in homogeneous coordinates are columns of the parity check matrix of $\C$.  Then   $S$ is a $(\rho,\mu)$-OS set with $\rho=d-2$.
\end{theorem}

\begin{proof}
    We use Proposition \ref{prop7}.
\end{proof}

By Theorem \ref{th7_R=d-1_2ccc}, we can construct optimal $(\rho,\mu)$-saturating sets for any $\rho$. Note that in the literature, see e.g. \cite{BDGMP_MultCov,BDGMP_MultCovFurth,BDMP_TwistCub,CHLL_CovCodBook},  $(\rho,\mu)$-OS sets are described only for $\rho\le2$.

\section{Conclusion. Open problems}
At the present time, the weight distributions of the cosets of MDS codes are investigated insufficiently. For non-binary
MDS codes,  sporadic or infinite families, the weight distributions of the code cosets (without any unions) are not considered, apart from the case of \cite{DMP_CosetsRScod4}.
In this paper some steps are made to fill this gap. In whole, the paper extends our knowledge on the problems connected with the weight distributions of the cosets of MDS codes.

We transformed the current best formula of \cite{Bonneau1990} for the coset weight distributions so that the new version is more convenient for applying and requires less calculations.

The new variant of the formula allowed us to make further simplifications for distinct specific cases connected with weights of the cosets. We considered separately cosets of weights $W$ with $W=1,2,d-1,d-2$, and $2\le W\le \lfloor(d-1)/2\rfloor$. In this consideration, for the cosets of weights $1$ and  $d-1$, we obtained formulas of the weight distributions depending only on the code parameters. This proves that all the MDS code cosets of weight $1$ (as well as $d-1$) have the same weight distribution.
The weight 2 and $d-2$ cosets may have different weight distributions; in this case, the distributions, as we proved, are symmetrical. A necessary condition for equality of the weight distributions of weight 2 cosets is established.

The weight distribution of the cosets of MDS codes corresponding to the conics and  regular hyperovals in the projective plane $\PG(2,q)$ was also considered. New properties of shortened conics, needed for the coset weight distributions, are obtained.

New results on multiple coverings of the farthest-off points (deep holes) called MCF-codes was obtained. The results can be considered as the classification of the deep holes. We showed that any MDS code of covering radius $R=d-1$ is an almost perfect MCF-code corresponding to an optimal multiple saturating set in the projective space $\PG(N,q)$. For the first time in the literature, almost perfect MCF-codes are constructed for arbitrary covering radii $R>3$.

Many interesting problems remain unresolved and could be investigated using approaches of this paper. As examples of the open problems, we note the following.

$\bullet$ Understand the combinatorial sense (in the context of
the problems considered) of the coefficients $\Omega_w^{(v)}(n,d)$ in \eqref{eq1_BonGenTransf} and obtain the formula directly without using the paper \cite{Bonneau1990}.

$\bullet$ Find the situations when the necessary condition of the equality of the weight distributions of weight 2 cosets is also sufficient, see Section \ref{sec_wd2}.

$\bullet$ Find the exact number of the weight $d-1$ cosets in the shortened GDRS codes of distance $d\ge5$, see Section \ref{sec_multip}.

$\bullet$ Find the values $c_i$ for infinite families of complete $n$-arcs, $n<q$, in the projective plane $\PG(2,q)$, see Section \ref{subsec:prelim_arc}.

$\bullet$ Investigate the cosets of codes corresponding to non-regular hyperovals in the projective plane $\PG(2,q)$, see Section \ref{subsec:prelim_arc}.

$\bullet$ Find the numbers $B_{d-2}(\V^{(W)})$ of weight $d-2$ vectors in the weight $W$ cosets of GDRS codes with $W=2,d-2$, see Sections \ref{sec_spec} and \ref{sec_wd2}.

\section*{Acknowledgments}   The research of  S. Marcugini and F.~Pambianco was
 supported in part by the Italian
National Group for Algebraic and Geometric Structures and their Applications (GNSAGA - INDAM) ((Contract No. U-UFMBAZ-2019-000160) and by
University of Perugia (Project No. 98751: Strutture Geometriche, Combinatoria e loro Applicazioni, Base Research Fund 2017-2019).

 \end{document}